\documentclass[10pt,journal,letterpaper,compsoc]{IEEEtran}
\usepackage{amssymb}

\usepackage{xcolor}
\usepackage{xkeyval}
\usepackage{tikz}
\usepackage{algorithmic}
\usepackage{algorithm}
\usepackage{graphicx}
\usepackage{stmaryrd}
\usepackage{amssymb,amsmath}
\usepackage{multirow,url}
\usepackage{color,cite}

\usepackage{graphicx}
\usepackage{stmaryrd}
\usepackage{amssymb,amsmath}
\usepackage{multirow,url}
\usepackage{color,cite}

\newtheorem{proposition}{Proposition}
\newtheorem{definition}{Definition}
\newtheorem{theorem}{Theorem}
\newtheorem{observation}{Observation}

\begin{document}

\title{{Motivating Smartphone Collaboration in Data Acquisition and Distributed Computing}}

\author{Lingjie~Duan,~\IEEEmembership{Member,~IEEE,} Takeshi~Kubo, Kohei~Sugiyama,~\IEEEmembership{Member,~IEEE}, Jianwei~Huang,~\IEEEmembership{Senior~Member,~IEEE,} Teruyuki Hasegawa, and Jean Walrand,~\IEEEmembership{Fellow,~IEEE}
\IEEEcompsocitemizethanks{\IEEEcompsocthanksitem Lingjie Duan$^1$ is with the Engineering Systems and Design Pillar, Singapore University of Technology and Design, Singapore-138682. Takeshi Kubo$^2$, Kohei Sugiyama$^3$ and Teruyuki Hasegawa$^4$ are with KDDI R\&D Laboratories Inc., Japan. Jianwei Huang$^5$ is with Department of Information Engineering, The Chinese University of Hong Kong. Jean Walrand is with Department of Electrical Engineering and Computer Sciences, University of California at Berkeley, California-94720. Email: $^1$lingjie\_duan@sutd.edu.sg, $^{2,3,5}$\{t-kubo,ko-sugiyama,teru\}@kddilabs.jp, $^4$jwhuang@ie.cuhk.edu.hk, $^5$wlr@eecs.berkeley.edu
}
\thanks{This work was supported by NICT of Japan, the General Research Funds (Project Number
412710) established under the University Grant Committee of Hong Kong,
and the CUHK Global Scholarship Programme for Research Excellence for
both junior faculty and PhD candidate. Also, the research of Jean Walrand was supported in part by NSF-NetSE grants 1024318 and 0910702, and the research of Lingjie Duan was supported by SUTD-MIT International Design Center (IDC) Grant (Project Number IDSF1200106OH).}
\thanks{Part of these
results were presented at IEEE INFOCOM 2012 \cite{INFOCOM2012}.}
}


\IEEEcompsoctitleabstractindextext{%
\begin{abstract}
This paper analyzes and compares different incentive
mechanisms {for a master to motivate} the collaboration of
smartphone users on both \emph{data acquisition} and
\emph{distributed computing} applications. To collect massive sensitive data from users, we propose a {reward}-based
collaboration mechanism, where the master announces a total reward
to be shared among collaborators, and the collaboration is
successful if
there are enough users wanting to collaborate.
We show that if the master knows the users' collaboration costs,
{then he can choose to involve only users with the lowest costs.} However, without knowing users' private information,
then he needs to offer a larger total reward to attract enough
collaborators. Users will benefit from knowing their costs before
the data acquisition. Perhaps surprisingly, the master may benefit as the variance of users' cost distribution increases.

{\quad \ To utilize smartphones' computation resources to solve complex computing problems,} we
{study how the master can design an optimal contract by
specifying} different task-reward combinations for different user
types. Under complete information, we show that the master involves
a user type as long as the master's preference characteristic
outweighs that {type}'s unit cost. {All
collaborators achieve a zero payoff in this case.} {If the
master does not know users' private cost information,} however,
{he} will conservatively target at {a
smaller group of users with} small costs, and has to
give most benefits to the collaborators.
\end{abstract}

\begin{keywords}
Smartphone application, data acquisition, distributed computing, game theory, contract theory
\end{keywords}}

\maketitle

\IEEEdisplaynotcompsoctitleabstractindextext

%
\IEEEpeerreviewmaketitle

\linespread{0.95}

\section{Introduction}
{Smartphones} are becoming the mainstream in
mobile phones. According to a survey by ComScore in 2010, over 45.5
million people owned smartphones out of 234 million total
{mobile phone} subscribers in the United States
\cite{gonsalves2010android}. In 2012, the global smartphone shipments grew 43\% annually by reaching a record 700 million units \cite{Epstein}.

Given millions of smartphones sold annually, {recent phone applications start to utilize the power of} smartphone users' collaborations
\cite{chen2000survey,Julia}. {In such an application, there is  a \emph{master} (e.g., Apple or Google in the following examples) who wants to implement some application or service based on user collaborations.}
%
We can categorize these applications in two types as follows.

In the first type of {\emph{data acquisition}} application, a master wants to acquire enough
data from smartphone users to build up a database. According to
\cite{Julia}, Apple's iPhone and Google's Android smartphones
regularly transmit their {owners}' location data (including GPS
coordinates) back to Apple and Google without users' agreements, respectively. For example, an
Android phone collects its location data every few seconds and
transmits the data to Google at least several times an hour. The
phone also transmits back the name, location, and signal strength of any
nearby Wi-Fi networks. After collecting enough location data from
users, Google can successfully build a massive database capable of
providing location-based services.
One service can be live map of {auto} traffics, where {the dynamics of} users'
location data on a highway indicate whether there is a
traffic jam. {Another service can be constructing a large-scale public
Wi-Fi map.}
%
According to \cite{2015}, the global location-based service market
is growing strongly, and its revenue is expected to increase from
US\$2.8 billions  to US\$10.3 billions between 2010 and 2015.
%
%
%
{In order to perform the above data acquisition, a lot of
efforts need to be spent to get users' consent and protect users'
privacy (e.g.,
\cite{Investigation,Senators,chow2006peer,gruteser2003anonymous}).
When a user collaborates {in this kind of applications}, he will incur a cost such as
loss of privacy. }

In the second type of {\emph{distributed computing}}
application, a master wants to solve complex engineering or
commercial problems inexpensively {using distributed computation
power}. Smartphones now have powerful {and power-efficient processors (e.g., Dual-core A6 chip of Apple iPhone 5 which is comparable to many laptops' CPUs several years ago)}, {outstanding battery life,} abundant memory, and open operating
systems (e.g., Google Android) \cite{TMT} that make them suitable
for complex processing tasks.
Since millions of smartphones remain unused most of the time, a
master might want to solicit smartphone collaborations in
distributed computing (e.g.,
\cite{rodriguez2011introducing,kelenyi2008energy,palmer2009ibis}). In this case, a
user's collaboration cost may be due to {loss of energy
and {reduction of} {physical} storage}.

{In this paper, we will design incentive mechanism for smartphone collaborations}
in data
acquisition and distributed computing applications, {both of which aim to incentivize users to participate the collaboration through proper rewards}. {Then we can compare the similarity and difference in mechanism design for both applications.}
{For each type of applications, we will similarly consider various information {scenarios},  depending on what the master and users know.}
{In particular}, the master may or may not know each {smartphone} user's characteristics
such as collaboration costs and collaboration efficiencies.

{The two types of applications have different requirements and
lead to different models.} {Collaborators in data acquisition usually take similar tasks and hence  should be rewarded similarly, whereas collaborators in distributed computing will undertake different amounts of work according to their different computation capabilities.} More specifically, in data acquisition applications, {we
consider a threshold-based revenue model, where a master can earn a fixed
positive revenue only if he can involve enough (larger than a
threshold) smartphone users as collaborators, such that he can
build a large enough database to support the application like the live map of auto traffics.} Since
data acquisition only requires {simple} periodic data reporting,
{we can assume that users} are homogeneous in contribution and
efficiency.\footnote{{For example, a huge number users take the same simple task by periodically reporting their GPS location data, and it is reasonable and fair for the master to reward them equally (as in Amazon Mechanical Turk). It is actually difficult for the master to differentiate contributions and rewards differentially to a huge number of users, and monitoring and updating his beliefs of users' private information is often impractical.}} In distributed computing applications, however, we consider
{a model where} the master's revenue increases in {users'
efforts}. Also, users {are} {heterogeneous {in computing efficiencies} and}
{should be treated differently. For example,  the most efficient users should be highly rewarded to encourage them to undertake large tasks.}

Our key results and contributions are as follows:
\begin{itemize}
\item \emph{New {reward}-based mechanism to motivate data
acquisition:} {We propose a Stackelberg game model {under incomplete information} that captures
interactions between the master and users in Section{~\ref{sec:Data}}. The
master first announces the total reward to be allocated among
collaborators.
To decide to join or not, each user then estimates other users' decisions in
predicting {the chance of collaboration
success} and his expected reward. {We show that it is better to reward users' collaboration efforts regardless the result of the collaboration. This encourages users to collaborate, and hence increases the chance of collaboration success.}}

\item \emph{Performance of reward-based mechanism:}
{Under complete or symmetrically incomplete information, the master can decide a small reward to attract enough users.} But if users {know
their costs} {while the master does not (asymmetrically
incomplete information)}, the master has to offer a large total
reward to {guarantee enough collaborators}, and users benefit from
{holding private information}. {Perhaps surprisingly, {when the master does not need a large number of collaborators}, he can benefit as the variance of users' cost distribution increases.}

\item \emph{New contract-based mechanism to motivate distributed
computing:} In Section~\ref{sec:compute}, we use contract theory to
study how a master {efficiently} decides different task-reward
combinations for heterogeneous users. {By satisfying individual rationality and incentive compatibility, our contract enables all users to truthfully reveal their private information and maximizes the master's utility.}

\item \emph{Performance of contract-based mechanism:} Under complete information, the
master {involves a user type as long as the master's preference
of the type is larger than {the user cost}. All collaborators
get a zero payoff. {{But if users can hold their private information from the master}, the master will conservatively
target at a smaller group of efficient users with small costs. He has to give most benefits to the collaborators and a collaborator's payoff increases in the computing efficiency.}}

\end{itemize}

\subsection{Related Work}
Our first collaboration model on data acquisition is closely related to the literature on
location-based services (LBS)\cite{rao2003evolution}. In LBS, a
customer needs to report his current location to the database
server in order to receive his desired service. Prior work are
focusing on how to manage data and how customers can safely communicate
with the database
server (e.g., \cite{schiller2004location,chow2006peer,gruteser2003anonymous}),
especially when the massive database has already been built
up. {Other work considered the technical issues of data collection from
users\cite{schiller2004location}.} {Our paper focuses on the master's problem of incentive mechanism design for attracting enough users ({larger than some threshold}) using reward to provide location data, so that the master can build a LBS later on.} {{Only recently people started to look at users' incentives to reveal information}. For example, Yang \emph{et al.} \cite{IncentiveMechanism} also designed incentive mechanisms for involving sensors. The model in \cite{IncentiveMechanism} does not involve the issue of collaboration cost estimation and collaboration  success probability, and the main results were mainly derived through simulations.} {Actually, most Stackelberg game models assume complete information yet this paper focuses on incomplete information.}


Our second collaboration model is relevant to mobile grid computing,
{which integrates} mobile wireless devices into grid computing
(e.g.,
\cite{phan2002challenge,wehner2010mobile,litke2004mobile}).
The {main focus of {mobile grid computing} literature} is on the technical
issues of resource management or load balancing (e.g.,
\cite{wehner2010mobile,litke2004mobile}). Only few results have
considered (mobile) users' incentives issues in joining in
collaboration
\cite{sim2006survey,kwok2007selfish,subrata2008cooperative}.
Kwok \emph{et al.} in \cite{kwok2007selfish} {evaluate the
impact of selfish behaviors of individual users in a Grid.} Subrata
\emph{et al.} in \cite{subrata2008cooperative} {present a Nash
bargaining solution for load balancing among multiple masters.}
Sim in
\cite{sim2006survey}  use a two-player alternating bargaining model
to study collaboration between masters and users. The novelty of our
{model} is that a master interacts with all users
simultaneously {to distribute computing work}, and users are heterogeneous in their computing
efficiencies and costs. {We propose a new contract{-based}
mechanism that maximizes the master's profit.}\footnote{{Our designed contract-based mechanism here {belongs to screening contract category \cite{bolton2005contract}} and is similar to that in our previous work \cite{DuanDySPAN} in methodology, but
that work focuses on a different problem on cooperative spectrum sharing and the derived mechanisms as well as results are significantly different.}} {Unlike our first collaboration model using a total reward to incentivize the same periodic data acquisition, here this contract mechanism aims to assign \emph{different} amounts of work to different user types, and does not require a threshold-based collaboration success.}


\section{{{Collaborations} on Data Acquisition}}\label{sec:Data}
\subsection{System Model of Data
Acquisition}\label{subsec:network_data}

In this application, the master is interested in building up a
database by collecting information f{rom enough smartphone
users. {We consider} a set $\mathcal{N}=\{1,\cdots,N\}$ of
smartphones, and the {total number} $N$ is publicly known.}\footnote{We assume all $N$ users are active. The master
(e.g., Apple) learns the number of active
{users (e.g., iPhones)} by checking {{users'}
usage history, or send users control messages for
status confirmation.} }
{User $i\in\mathcal{N}$ has a collaboration cost {$C_i(p_i,e_i)>0$, which is generally a function (e.g., weighted sum) of his privacy loss $p_i$ and energy consumption $e_i$ illustrated as follows}:
\begin{itemize}
\item \emph{Privacy loss $p_i$:}
By reporting sensitive data (e.g., GPS location coordinates), a user's loss can be psychological worry of losing privacy, discomfort due to frequent annoyance from unwanted advertising in location-based services, or even property loss due to disclosure of bank account information in data reporting (e.g., \cite{karger1995security, langheinrich2001privacy, want1992active}).
\item {\emph{Energy consumption $e_i$:}  Collecting and transmitting data periodically to the master's data center consumes a user's smartphone battery. According to \cite{qian2011profiling} and \cite{pathak2012energy}, the consumed energy depends on the details of the data acquisition task, including the interaction efficiency among various layers (e.g., radio channel state, transport layer, application layer, and user interaction layer). The measurement data (e.g., radio power) for some typical applications and platforms can be found in \cite{qian2011profiling} and \cite{pathak2012energy}.}
\end{itemize}

{We assume that the distributions of users' privacy losses and energy consumptions are independent. As combinations of the two terms, we assume that the collaboration costs are independently and
identically distributed, with a mean {$\mu$} and a cumulative
probability distribution function {$F(\cdot)$}.\footnote{{This assumption makes our analysis tractable to deliver clean engineering insights. Yet our results can be extended to the case where the costs are not identically distributed. Then, for example, a user with a larger mean value of cost distribution is less willing to collaborate.}} We do not impose any further assumptions on the properties of the distribution $F(\cdot)$ in this paper. }


We consider a \emph{threshold revenue model} for the master. {If the master attracts at least $ n_0$ users as collaborators, he will successfully build the database and receive a revenue of $V$. Otherwise, the master does not receive any revenue.} {Such a threshold model has many practical applications. In the example of collecting GPS data to establish a live map of auto traffics, several users' movement information along the same highway will be enough to tell whether the highway is congested or not.}

As shown in Fig.~\ref{fig:twostage}, the master interacts with the users through a two-stage process.
In Stage I, {the master announces $(R,n_{0}$), where $R$ is the
total reward to all users and $n_{0}$ is the {threshold number
of} required collaborators.} In Stage II, {each user chooses to be a collaborator or not.}\footnote{{To support some real-time location-based services (e.g., {maps of live traffic information}), we require users' quick responses {in order to update the database}. Such a quick {collection} of data is feasible now (e.g., like invitation messaging to reveal location in many iPhone apps).  Amazon Mechanical Turk also supports online interaction between masters and users. Some media masters like Google may also have {urgent} needs to report some critical events by asking users {in a certain area to upload photos and videos}.}} {Similar to Amazon Mechanical Turk,\footnote{See the Amazon link \url{https://www.mturk.com/mturk/welcome}} the master here sets up a database with users' account and payment information and can automatically pay each involved user. A user's received reward can be monetary return or some promotion to use the relevant location-based service afterwards.}

{
We want to mention that our model is also applicable to the scenario where users do not join the collaboration simultaneously. As long as users submit their ``sealed'' responses to the master's collaboration invitation and cannot check others' behaviors, our results will remain valid. On the other hand, if a user learns from others' behaviors to determine the current number of committed users, then the analysis of such a dynamic decision evolution becomes very challenging. Some methods about social learning and mean field approximation may be used in the scenario (though clean theoretical results are still hard to obtain) (\!\!\cite{ifrach2012monopoly,weintraub2010computational}).
}

\colorlet{lightgreen}{green!20!yellow} \colorlet{lightblue}{blue!20}
\colorlet{lightred}{red!20}
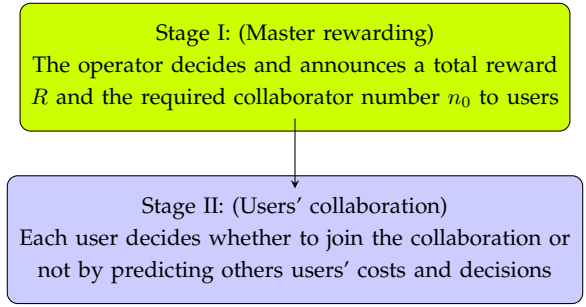
\begin{figure}[tt] 
   \centering
    \begin{tikzpicture}[>=stealth, scale=0.85]
        \draw[fill=lightgreen, rounded corners, transform shape] (-4.3,3) rectangle (4.3,5);
        \draw[fill=lightblue, rounded corners, transform shape] (-4.5,0.3) rectangle (4.5,2.3);
        \draw (0,4.5) node [fill=lightgreen, rounded corners, transform shape] {Stage I: (Master rewarding)};
        \draw (0,4) node [fill=lightgreen, rounded corners, transform shape] {The operator decides and announces a total reward};
        \draw (0,3.5) node (Price) [fill=lightgreen, rounded corners, transform shape] {$R$ and the required collaborator number $n_0$ to users};

        \draw (0,1.8) node (Users') [fill=lightblue, rounded corners, transform shape] {Stage II: (Users' collaboration)};
        \draw (0,1.3) node [fill=lightblue, rounded corners, transform shape] {Each user decides whether to join the collaboration or};
        \draw (0,0.8) node [fill=lightblue, rounded corners, transform shape] {not by predicting others users' costs and decisions};
        \draw[->] (Price.south) -- (Users'.north);
        \end{tikzpicture}
    \caption{\small Stackelberg game between the master and users.} \label{fig:twostage}
\end{figure}

{Assume that there are $n$ out of $N$ users willing to serve as collaborators in Stage II.} There are two models for a collaborator's payoff:
\begin{itemize}
\item \emph{Model (A) (Reward for collaboration effort):} A collaborator $i$'s payoff is
\begin{equation} \label{eq:up2}
\left(\frac{R}{n} - C_i\right) \boldsymbol{1}{_{\{n \geq n_0\}}}.
\end{equation}
where {$\boldsymbol{1}_{\{X\}}$} is the indicator function {(equals 1 when event $X$ is true)}.
{That is, if the collaboration is successful, user $i$ pays his
collaboration cost $C_{i}$, and gets the reward $R/n$ (equally and fairly
shared among $n$ collaborators since they undertake the same task in fixed and periodic data reporting). We can also view $R/n$ as in a \emph{lottery} scenario where each collaborator having equal probability $1/n$ to win the total reward $R$. In this case, $n$ users will
only collaborate if the master notifies them that $n\geq
n_{0}$ and the collaboration will be successful. This means that no users will pay collaboration cost if the collaboration is not successful. Here, we assume
that the master will truthfully inform the collaborators about the
value of $n$.\footnote{In reality, the master may cheat users by announcing a larger value of $n$, then he can give less reward to each actual collaborator. But there are some approaches to prevent this. {For example, there could be a third party to monitor how many collaborators are finally involved and punish the master if cheating is detected.}}}
\item \emph{Model (B) (Reward only with successful collaboration):} A collaborator $i$ receives a payoff
\begin{equation}\label{eq:down2}
\frac{R}{n}\boldsymbol{1}_{\{n \geq n_0\}}-C_i.
\end{equation}
That is, collaborator $i$ always pays his collaboration cost $C_i$, and will get the reward $R/n$ only if the collaboration is successful. This model considers that collaborators will contribute before they know the value of $n$ (which will be announced to them by the master after data acquisition).
\end{itemize}



{In both model, the master obtains a profit of}
\[
(V - R)\boldsymbol{1}_{\{n \geq n_0\}}.
\]



For illustration purpose, we now only focus on Model (A) in this section. The discussion of Model (B) can be found in Appendix~{F}. It should be noted that users under Model (A) are more
willing to collaborate than under Model (B), which is not surprising since they face a lower risk in Model
(A). The master also prefers Model (A) to Model (B) since he needs to compensate lower risk and fewer cost to motivate users' collaboration.

The {\em collaboration game} is a two-stage Stackelberg
game, {and we would like to characterize the subgame perfect equilibrium (SPE) that specifies players' stable choices in all stages \cite{gibbons1992game}.} {The way to analyze Stackelberg game
is backward induction. {The master in Stage I and users in Stage II are risk-neutral and want to maximize their own payoffs, repectively.} We will first analyze Stage II, where the
users play a game among themselves based on the value of the reward
$R$ and the threshold $n_{0}$.\footnote{We consider that each user
will join the collaboration as long as his payoff is nonnegative. Yet our results can still be generalized to
the case where users have positive reserve payoffs.}
Users reach a \emph{Nash equilibrium} {(NE)} in this stage, if no user can improve
his payoff by changing his strategy {(collaborate or not)} unilaterally. The Nash equilibrium in
{Stage II leads to}  a collaboration success probability $P(n
\geq n_0| R)$. As we will see, there may be multiple Nash equilibria
in Stage II. Then we study Stage I, where the master chooses the
value of $R$ to maximize his expected profit $(V - R)P(n \geq n_0|
R)$. These two-step analysis {enables us} to {obtain} an SPE of the whole two-stage
collaboration game.

Next we {will analyze the Stackelberg game, and }{study
how the master's {and the users'} information about the collaboration costs
will affect the outcome.}
}

\vspace{-5pt}
\subsection{Collaboration under Complete
Information}\label{subsec:symm_data}

{{We first consider the complete information scenario, where the
master and all users know the cost $C_i$ of every user
$i\in\mathcal{N}$.} {{This is possible only in some special cases where the master and users have extensive prior collaboration experiences.} The main reason for studying this model is to provide a performance benchmark for later discussions of  more realistic incomplete information scenarios. }

{{We assume that no two users have
the exactly same cost. {Our results also apply to the case of homogeneous users, where we can randomly break the tie among homogeneous users at the boundary. This will lead to more than one equilibrium.}}} Without loss of generality, we reorder the users' costs in ascending order, i.e., $C_1<C_2<...<C_N$ and $C_{n_0}$ is the $n_0$th smallest cost. The {equilibrium} of the collaboration game
is as follows.

\begin{theorem}[Collaboration under Complete Information]\label{thm:data_complete}
Recall that $C_{n_0}$ is the $n_0$-th smallest collaboration cost among all $N$
users. The collaboration game admits the following unique \emph{pure} strategy
SPE.
\begin{itemize}
\item If $V < n_0C_{n_0}$, then {the master does not want to initiate the collaboration in Stage I and sets $R^*=0$.  No user will become collaborator in Stage II.}

\item {If $V \geq n_0C_{n_0}$, the master offers a reward $R^*=n_0C_{n_0}$ in Stage I. In Stage II, every user $i$ with $C_i \leq C_{n_0}$ collaborates and obtains a {nonnegative} payoff $C_{n_0} - C_i$, and the remaining $N-n_0$ users decline to collaborate and get a zero payoff. The profit of the master is $V - n_0C_{n_0}$.}
\end{itemize}

\end{theorem}

%
%

The proof of Theorem~\ref{thm:data_complete} is given in Appendix~{A}.

We can show that users will not benefit from using a
mixed-strategy.
{But this may not be the case with
symmetrically incomplete information.}

\vspace{-5pt}
\subsection{{Collaboration under {Symmetrically} Incomplete Information}}\label{subsec:noinfo_data}

Now we consider the symmetrically incomplete information scenario,
where {both} the master and the users only know the cumulative probability
distribution function $F(\cdot)$ of the collaboration costs with the mean $\mu$.\footnote{{The master can estimate $F(\cdot)$ by learning from his collaboration history or making a customer survey. A user can estimate $F(\cdot)$ by checking his or other users' collaboration experiences. There are many public sources (e.g., the master's or some third party's market or customer surveys) that help a user's cost estimation \cite{schiller2004location,2015}.}} {A
user $i$ even does not know the precise value of his own cost
$C_{i}$.\footnote{{It is sometimes difficult for a user to know
his precise loss of privacy before an actual security threat happens to him.} {Users may face many possible security threats by losing sensitive information, e.g., direct property loss or advertising harassment.}} {In this case, we can view all users as homogeneous.}}

\vspace{-5pt}
\subsubsection{Analysis of Stage II}

It turns out that there are multiple {equilibria} of the
collaboration game in Stage II as follows.

\begin{theorem}\label{them2}\emph{(Stage II under {Symmetrically} Incomplete
Information):} Stage II admits the following Nash equilibria:\
\begin{itemize}
{ {\item \emph{(No collaboration):} If $R<n_0\mu$, no user
will collaborate at any equilibrium in Stage II.}
\item\emph{(Pure strategy NE):} If $n_0\mu \leq R<N\mu$ {(where $N\mu$ is the product of user number $N$ and the mean $\mu$ of user cost distribution)}, $n^*=\lfloor
\frac{R}{\mu}\rfloor$ users choose to collaborate and the remaining users decline.
A subset of $n^*$-out-of-$N$ users is randomly picked up
among $\left(\begin{array}{c}N \\ n^*\end{array}\right)$ possible
subsets. {Thus there exist multiple pure NEs in this
case.}\footnote{How to select one NE is out of the scope of this paper, and can be referred to \cite{gibbons1992game}. {Yet it should be noted that all pure strategy Nash equilibria lead to the same performance for the master as shown in Theorem~\ref{thm:data_symm} later on.}} If
{$R\geq N\mu$}, all $N$ users will collaborate.

\item \emph{(Mixed strategy NE):} If $n_0\mu<R<N\mu$, every user
collaborates with a probability $p^*$, which is the unique solution to
\begin{equation}
\mathbb{E}_{m}\left(\left(\frac{R}{m+1}-\mu\right)\boldsymbol{1}_{\{m+1\geq
n_0\}}\right)=0,
\end{equation}
{where the expectation $\mathbb{E}$ is taken over the random variable $m$ which follows a binomial}
distribution $B(N-1,p)$.}
%
%
%

\label{thm:StageII_NoInfo}
\end{itemize}


\end{theorem}


The proof of no collaboration and pure strategy NE are given in Appendix~{B}.

{We note that the pure and mixed strategy equilibria in Theorem
\ref{them2} {share a common parameter range,} $n_{0}\mu<R<N\mu$.}
{It should also be noted that the master is not interested in selecting a certain NE since all NEs give him the same performance. Furthermore, Theorem~\ref{thm:data_symm} will show that the master will not encourage any mixed NE at the first place.
}




{

Next we show how the \emph{mixed strategy NE} $p^*$ is derived.} As
all users have the same {statistical} information, we will
focus on the symmetric mixed Nash equilibrium. Assume that {all
users} collaborate with a probability $p$. {Consider that there are $m$ users {(other than $i$)} who
collaborate.} If user $i$ collaborates,
his {expected} payoff is
\[
u (R, p) :=  \mathbb{E}_{m}\left(\left(\frac{R}{m + 1}  - \mu\right)
\boldsymbol{1}_{\{m + 1 \geq n_0\}}\right),
\]
where the expectation is taken over $m$ which
follows a binomial distribution $B(N-1,p)$, and is independent {of user $i$'s decision.}

{Given all the other $N-1$ users collaborate with the equilibrium probability $p^{\ast}$, user $i$'s payoffs by choosing to collaborate or not are the same. Thus $p^*$ should satisfy}
\[
u (R,{p}^*)=0,
\]
{and is} a function of $R$. Thus we can rewrite $p^*$ as $p^*(R)$. One
can show that there exists a mixed strategy Nash equilibrium
$p^*(R)\in(0,1)$ as long as $n_0 \mu<R<N\mu$. Note that $R\leq n_0\mu$ leads to $p^*(R)=0$, which is
not a mixed strategy. Also, $R\geq N\mu$ leads to $p^*(R)=1$, which is not a mixed strategy either.

\vspace{-5pt}
\subsubsection{Analysis of Stage I}

{First we consider the case where users use the mixed strategy
in Theorem \ref{them2} and collaborate with probability $p^*(R)$.}
The master's expected profit is then
\[
f(R) := \mathbb{E}_n\left( (V - R) \boldsymbol{1}_{\{n \geq
n_0\}}\right),
\]
where the expectation is taken over $n$ which follows a binomial
distribution $(N, p^*(R))$. One can show that $f(R)$ {has} a
unique maximum $f(R^*)$, which is positive when $V
> n_0 \mu$. {However, under {$n_{0}\mu<R<N\mu$} there is always a chance that there are less than $n_{0}$ users choosing to collaborate under the mixed strategy.}
Thus the
master may want to avoid this.  
Theorem~\ref{them2} shows that by choosing $R=n_0\mu$, the master
can guarantee $n_0$ collaborators with a pure strategy
Nash equilibrium in Stage II. {Any reward value lower than $n_0\mu$ leads to no collaboration, and any value larger than $n_0\mu$ involves a number of collaborators that is larger than necessary (in a pure strategy NE) or does not guarantee enough collaborators (in a mixed strategy NE). As the master's payoff decreases in reward given enough collaborators,} {we have the following
result.}


\begin{theorem}\label{thm:data_symm}\emph{({Stage I under Symmetrically Incomplete
Information}:)} The collaboration game admits the following unique
SPE.

\begin{itemize}
\item If $V<n_0\mu$, the master will not initiate the collaboration
and will choose $R^*=0$.

\item If $V \geq n_0 \mu$, the master will announce a reward
$R^*=n_0\mu$. {A set of} $n_0$ users will collaborate in Stage
II. The collaborators achieve a zero expected payoff, and the master achieves a
profit $V - n_0 \mu$.
\end{itemize}

\end{theorem}


\vspace{-5pt}
\subsection{Collaboration under Asymmetrically Incomplete Information}\label{subsec:asymm_data}
In this subsection, we study the case where each user $i$ knows his
own exact cost $C_i$, but not other users' costs. The
master only knows $F(\cdot)$.\footnote{{This is possible when users already learn about their private costs over time, yet the master may not be able to track and collect these sensitive information.}}}

\subsubsection{Analysis of Stage II}
We have the following result for Stage II.

\begin{theorem}\label{them4}\emph{(Stage II under Asymmetrically Incomplete Information):}
A user $i$ will collaborate if and only if $C_i\leq \gamma^*(R)$.
The {common equilibrium} decision threshold $\gamma^*(R)$ is the
unique solution of {$\Phi(\gamma)=0$}, where
\begin{equation}\label{eq:StageII_asymm}
\Phi(\gamma):= \mathbb{E}_{m}\left(\left(\frac{R}{m+1}-\gamma\right)
\boldsymbol{1}_{\{m+1 \geq n_0\}}\right),
\end{equation}
and the expectation is taken over $m$ which follows a binomial
distribution $B(N-1,F(\gamma))$. The equilibrium $\gamma^*(R)$
satisfies $\frac{R}{N}<\gamma^*(R)<\frac{R}{n_0}$.\footnote{{We can
also show that a user will not be better off by changing from the
current pure strategy to any mixed strategy.}}
%
\end{theorem}

\begin{figure}[t]
\centering
\includegraphics[width=0.41\textwidth]{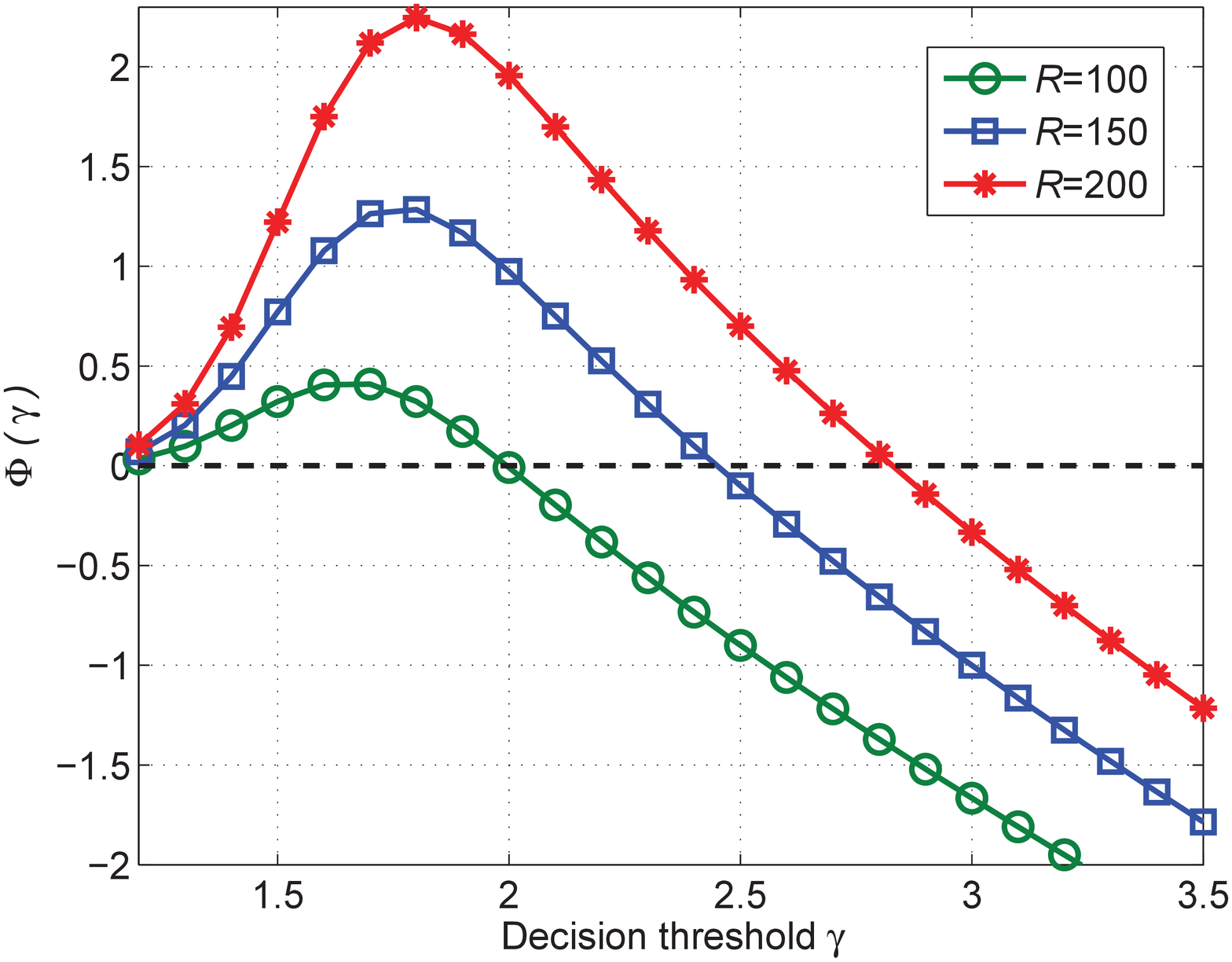}
\caption{\small $\Phi(\gamma)$ as a function of $\gamma$ and $R$. Other
parameters are $n_0=40$ and $N=100$. We consider a uniform cost
distribution with $F(\gamma)=\min(\gamma/4,1)$. {The threshold values $\gamma^*(R)$ for the three different values of $R$ (in ascending order) values are 2.0, 2.4, and 2.8, respectively.}} \label{fig:Phi_rho}
\end{figure}

{As (\ref{eq:StageII_asymm}) is independent of $C_i$ when each user $i$ makes his decision, all users have the same decision threshold. The intuition is that each user has the same information and estimation about others. But different users would still make different decisions as their private information about their own collaboration costs are different.}

To see why Stage II has the pure NE in Theorem~\ref{them4}, we consider
that all users
other than $i$ collaborate if and only if their costs are less than
some threshold $\gamma> 0$.
If user $i$
collaborates, his payoff is
\[
\left(\frac{R}{m+1}-C_i\right)\boldsymbol{1}_{\{m+1\geq n_0\}},
\]
where $m$ follows a binomial distribution $B(N - 1, F(\gamma))$ and
represents the number of users (other than $i$) who collaborate.
(Recall that cdf $F(\gamma) = P(C_i \leq \gamma)$.) Accordingly, the
expected payoff of user $i$ if he collaborates is
\begin{equation}\label{Eq:exptedprofitD}
\mathbb{E}_{m}\left(\left(\frac{R}{m+1}-C_i\right)\boldsymbol{1}_{\{m
+ 1 \geq n_0\}}\right),
\end{equation}
and zero otherwise. {At the Nash equilibrium, (\ref{Eq:exptedprofitD}) should equal to $0$ when $C_i =
\gamma$. That is, having the common collaboration threshold $\gamma$ is a
Nash equilibrium if and only if} $\Phi(\gamma) = 0$.
{We denote the
solution to $\Phi(\gamma)=0$ in (\ref{eq:StageII_asymm}) as $\gamma^*(R)$. In Appendix~{C}, we prove that
there always exists a unique $\gamma^*(R)$, which satisfies $\frac{R}{N}<\gamma^*(R)<\frac{R}{n_0}$.}

%
%

Figure~\ref{fig:Phi_rho} shows $\Phi(\gamma)$ as a function of both
$\gamma$ and $R$. The solution $\gamma^*(R)$ to $\Phi(\gamma)=0$ is
always unique and satisfies
$\frac{R}{N}<\gamma^*(R)<\frac{R}{n_0}$.
When $R=100$, for example, we have $\gamma^*(R)=2$, which
is larger than $R/N=1$ and is smaller than
$R/n_0=2.5$. {It is also interesting to notice that all users share the same decision threshold $\gamma^{\ast}(R)$ although they have different costs.}

\begin{theorem}\label{thm:gamma_R_n0_N}
The equilibrium decision threshold $\gamma^*(R)$ increases in $R$,
and decreases in $N$ and $n_0$.
\end{theorem}
The proof of Theorem~\ref{thm:gamma_R_n0_N} is given in Appendix~{D}.

\begin{figure}[t]
\centering
\includegraphics[width=0.41\textwidth]{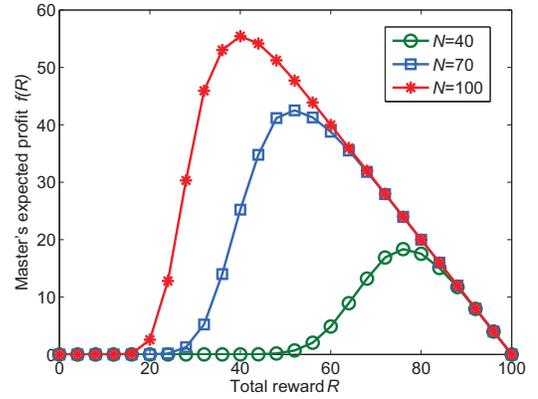}
\caption{\small Master's expected profit $f(R)$ as a function of $R$ and
$N$. Other parameters are $n_0=30$ and $V=100$. Also, we consider a
uniform cost distribution with $F(\gamma)=\min(\gamma/3,1)$. {The optimal reward values for different $N$ values (in ascending order) are 77, 52, and 40, respectively.}}
\label{fig:pi_R_N}
\end{figure}

As $N$ or $n_0$ increases, more users need to
{participate in} the collaboration and {thus the shared
reward per collaborator decreases. Therefore, the decision
threshold decreases and each user is less likely to collaborate.}



\subsubsection{Analysis of Stage I}
We are now ready to consider Stage I. {Given users' equilibrium
strategies based on threshold $\gamma^\ast(R)$ in Stage II in Theorem~\ref{them4}}, the master {chooses{ reward $R$} to maximize his expected profit, i.e., }{
\begin{equation}\label{eq:asymm_f}
\max_Rf(R) = \mathbb{E}_n\left((V - R)\boldsymbol{1}_{\{n \geq
n_0\}}\right),
\end{equation}
where the expectation is taken over $n$ which follows {a
binomial distribution $B(N, F(\gamma^*(R)))$}.}
{A smaller reward $R$ leads to
a larger value of $V-R$, but decreases the collaboration success
probability $P(n\geq n_0;R)$.}

{Let us denote the master's equilibrium choice of reward in Stage I as $R^{\ast}${, which is the optimal solution to Problem~(\ref{eq:asymm_f}).}} {To solve Problem~(\ref{eq:asymm_f}), we can use any one-dimensional exhaustive search algorithm to find the global optimal solution.} {Next we verify that the computation complexity is not high. We can approximate the continuity of the feasible range $[0, V]$ of reward $R$ through a proper discretization, i.e., representing all possibilities of $R$ by $\bar{V}$ equally spaced values (with the first and last values equal to 0 and $V$, respectively). Since the user decision threshold $\gamma ^*(R)$ for each possible value of reward $R$ belongs to the range $[0, V]$ (due to $\gamma ^*(R)\leq R$ and $R\leq V$), we can similarly discretize this continuous range of $\gamma ^*(R)$ by $\bar{V}$ possible values. To derive $\gamma^*(R)$ for each possible $R$ value, we need to search over all $\bar{V}$ possibiliites of $\gamma^*(R)$ to apprximately solve $\Phi=0$ in (\ref{eq:StageII_asymm}), and to derive the optimal $R^*$ we need to further search over all $\bar{V}$ possilities of $R$, and thus the overall computation complexity to solve Problem (\ref{eq:asymm_f}) is $\mathcal{O}(\bar{V}^2)$.} {The choice of $\bar{V}$ depends on the master's tolerance level of the quantization error, and a larger $\bar{V}$ value leads to a more accurate solution with more computation overhead.}

\begin{theorem}\label{thm:data_asymm_f}
The {equilibrium} expected
profit $f(R^{\ast})$ of the master increases in $V$ and $N$, and
decreases in $n_0$.

\end{theorem}

The proof of Theorem~\ref{thm:data_asymm_f} is given in Appendix~{E}. {Similarly, we can show that the optimal reward $R^*$ increases in $V$ and $n_0$, and decreases in $N$.}

\begin{figure}[t]
\centering
\includegraphics[width=0.41\textwidth]{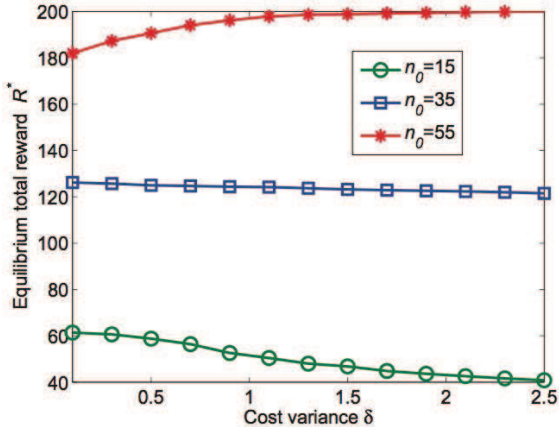}
\caption{\small Master's equilibrium reward $R^*$ as a function of variance $\delta$ and $n_0$. Other parameters are $N=80$, $V=210$, and $\mu=3$. }
\label{fig:Reward_Variance}
\end{figure}

As the master's revenue $V$ increases, he benefits more from the
collaboration.  As the {threshold}  $n_0$ increases, {however,}
each user is less likely to collaborate. Thus the master has to give a
larger total reward to attract enough collaborators. This decreases
his equilibrium expected profit.

Figure~\ref{fig:pi_R_N} shows that the master's expected profit
$f(R)$ as a function of $R$ and $N$. We can see that both $f(R)$ and
the equilibrium $f(R^{\ast})=\max_{R}f(R)$ are increasing in $N$.
Intuitively, as $N$ increases, more users have small collaboration
costs {(as the cdf function $F(\cdot)$ does not change)}, and more
users will collaborate under the same total reward. Thus
the master can lower the equilibrium reward $R^*$ and obtain a
larger expected profit.

Next we study how the master's equilibrium total reward and expected profit change with the cdf function $F(\cdot)$ of a user's collaboration cost. We pick Gaussian distribution for example, which can be explicitly characterized by mean $\mu$ and variance $\delta$ only. This is the case where the master aggregates his cost observations over a large number of user samples.\footnote{Note that the following results also apply to uniform distribution, and we skip the discussion here due to the page limit.}

\begin{observation}
The master's equilibrium total reward $R^*$ increases in mean $\mu$ (and his expected profit decreases in $\mu$). Moreover, the optimal reward $R^*$ increases in variance $\delta$ for large $n_0$ and decreases in $\delta$ for small $n_0$.
\end{observation}

The relationship between $R^*$ and $\mu$ is quite intuitive, as the master needs to decide a larger $R^*$ to compensate each collaborator's increased cost in expected sense. We next elaborate the impact of $\delta$ on $R^*$.

Figure~\ref{fig:Reward_Variance} shows $R^*$ as a function of variance $\delta$ and $n_0$, where the master with smaller $n_0$ requirement can efficiently build the database with smaller reward $R^*$. When the master requires a large $n_0$, he needs to incentivize most users to join the collaboration. As $\delta$ increases, some users are more likely to realize much larger costs than $\mu$ and the conservative master still needs to incentivize them. Thus $R^*$ increases in $\delta$ in this case. When the master only requires a small $n_0$, he can target at those users with smallest costs. As $\delta$ increases, these users are more likely to have much smaller costs than $\mu$ and the master only needs to decide smaller $R^*$ to incentivize them.

We can similarly show in Fig.~\ref{fig:Profit_Variance} that as $\delta$ increases, the master with large $n_0$ requirement obtains smaller expected profit $f(R^*)$, whereas the master with small $n_0$ requirement obtains larger $f(R^*)$. Notice that $f(R^*)$ decreases in requirement $n_0$, which is consistent with Theorem~\ref{thm:data_asymm_f}.

\begin{figure}[t]
\centering
\includegraphics[width=0.41\textwidth]{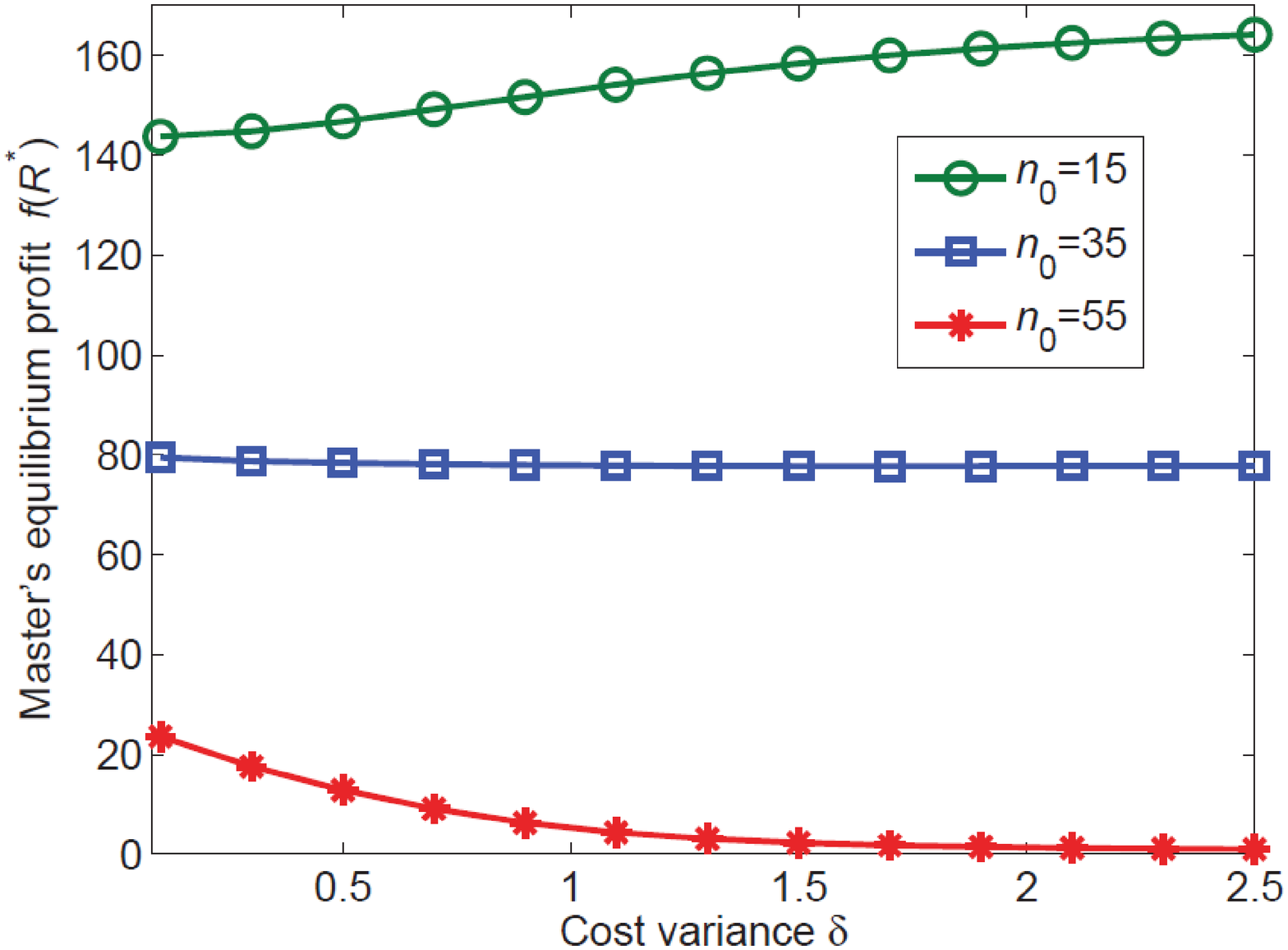}
\caption{\small Master's equilibrium expected profit $f(R^*)$ as a function of variance $\delta$ and $n_0$. Other parameters are $N=80$, $V=210$ and $\mu=3$. }
\label{fig:Profit_Variance}
\end{figure}

\begin{figure}[t]
\centering
\includegraphics[width=0.4\textwidth]{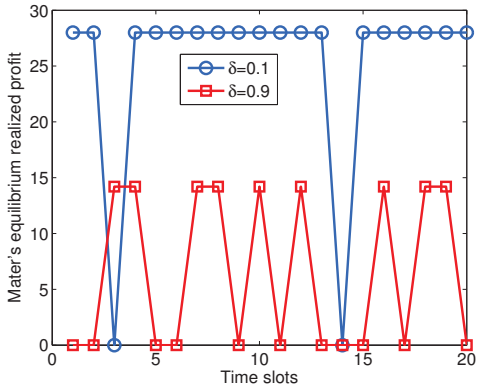}
\caption{\small Master's equilibrium realized profit in time horizon. Other parameters are $N=80$, $V=210$, $n_0=55$ and $\mu=3$. }
\label{fig:Profit_Variance_Realized}
\end{figure}


As the users' costs {are random variables and} have different realizations {in different time slots}, we {explore} how the master's equilibrium \textbf{realized} profit changes with time and cost variance $\delta$ in the newly added Fig.~\ref{fig:Profit_Variance_Realized} when {the contributor threshold} $n_0=55$. {Notice that} for each $\delta$ value, the optimal reward $R^*$ is determined to maximize \emph{expected} profit {based on users' cost distributions}, {and does not depend on the cost realizations in each time slot}. The realized profit is either $V-R^*$ or $0$, depending on {whether the collaboration is successful in that time slot} (i.e., $n\geq n_0$) according to (6). As $\delta$ increases, the master's realized profit in Fig.~6 decreases most of the time (except time slots 3 and 14). {This is consistent with the result in Fig.~\ref{fig:Profit_Variance} regarding the equilibrium expected profit when $n_{0}=55$.} Furthermore, Fig.~\ref{fig:Profit_Variance_Realized} shows {that} a larger variation $\delta$ leads to more fluctuations of the mater's realized profit. {This means that} when $n_0$ is large, the master prefers a smaller cost variation {in order to have a larger realized profit in most of the time slots}. {When $n_0$ is small (e.g., $n_0=15$), we can show that (through a figure of a similar style as Fig.~\ref{fig:Profit_Variance_Realized}) that} the master prefers a larger cost variation instead. {We do not include the small $n_{0}$ case here due to the space limit.}

By comparing the performances of the master and users under complete
information, {symmetrically incomplete information}, and
asymmetrically incomplete information, we have the following result.

\begin{theorem}\label{ob:data_symm_asymm}
At the equilibrium of the collaboration game, the master obtains the
smallest expected profit under \emph{asymmetrically} incomplete
information, whereas the users obtain the smallest (zero)
expected payoffs under {\emph{symmetrically} incomplete information}.
\end{theorem}

{Theorem \ref{ob:data_symm_asymm} shows that the users benefit
from knowing their own costs, while the master incurs profit loss
when the users know their costs and can hide the information from the master. }

Recall that the master obtains an expected profit $V-n_0\mu$ under
{symmetrically incomplete information}, and obtains a profit
$V-n_0C_0$ under complete information. The relation between these
two values depends on $N$, $n_0$, and $F(\cdot)$. Take the uniform
distribution $F(\cdot)$ as an example. If $n_0$ is much smaller than
$N/2$, the expected value of $C_0$ will be smaller than $\mu$ and
the master is better off under complete information.

\section{{Collaborations on Distributed {Computing}}}\label{sec:compute}
\subsection{System Model on Distributed
Computation}\label{subsec:network_comput}

In this type of  applications, the master solicits the collaboration
of smartphone users to perform distributed {computing}.\footnote{{PCs are also suitable to handle distributed computing given a lot more power. Our results here can also be applicable for those PCs which are underutilized and can connect to Internet for networking.}}
{Different from requiring fixed and periodic data reporting} as
in {data acquisition applications}, the master here can assign different
amounts of work to different user types. Smartphones are generally
different in terms of CPU performance, {memory and storage},
battery life, and
{connectivity}\cite{phan2002challenge}. {Even with the same smartphones, two users may have different phone usage behaviors and different sensitivities (e.g., energy consumption).}


{One can imagine that the energy consumption will hinder the smartphones' involvement in distributed computing, {which motivates us to} consider smartphones' energy constraints (i.e., battery capacity limits) in the modeling of $\bar{t}_i$ in (8). {Meanwhile,} we want to highlight that {it is already feasible for} smartphones {to support} distributed computing. First, the large {number} of smartphones can help compensate {the energy limitation of each individual phone} \cite{rodriguez2011introducing}. Then the energy consumed by an individual smartphone is not large. Second, the energy limitation will {be of a less concern for smartphones which have access to charging facilities (as the reviewer has pointed out).}
With {the newly developed} wireless charging {technologies} (e.g., inductive and magnetic resonance couplings), more smartphones can be supported even when they are moving around \cite{graham2009wireless}. Third, the battery technologies and {energy management algorithms have been significantly improved during the recently years} (e.g., \cite{kelenyi2008energy,palmer2009ibis}). {Finally, today's} data storage technologies make it possible to store tens of gigabytes {of data} in a small {memory} card, which {means that some of today's smartphones are almost as capable as desktop computers from several years back} \cite{rodriguez2012energy}.}

We consider a total of $N$ users belonging to {a set
$\mathcal{I}=\{1,\cdots,I\}$ of} $I$ types. {Each type
$i\in\mathcal{I}$ has $N_{i}\geq 1$ users, with
$\sum_{i\in\mathcal{I}} N_{i}=N$.} A type-$i$ user {can perform
at most $\bar{t}_i$ units of work}, and faces a cost $K_i$ per unit
of work he performs. The upper bound
{of $\bar{t}_i$} reflects the limited battery capacity, time
constraint, or other physical constraints. {\label{footnote:techspec}Users know
their unit costs before the collaboration, {since (i) many factors of these costs (e.g.,
power consumption) are explicitly reflected by smartphones'
technical specifications, and (ii) users explicitly know their own sensitivities (e.g., to power consumption) in costs.} {Note that the data exchanged between users and the master are not sensitive to users.} This is different from data acquisition,
where costs {also} come from implicit insecurity.} {To determine the unit cost values, one can check \cite{qian2011profiling} and \cite{pathak2012energy} for energy consumption data (e.g., radio power in joul per minute) in specific applications.}

The payoff of a
type-$i$ user who accomplishes $t$ units of work and receives a
reward $r$ from the master is
\begin{equation}\label{eq:u_i}
u_i(r, t)= r-K_it, \mbox{ for } 0 \leq t \leq \bar{t}_i.
\end{equation}
Note that the user can always choose not to
collaborate with the master and thus receive zero payoff with
$t=r=0$.\footnote{{Note that rewards and tasks are related in a properly designed contract, and one would expect that the reward increases with the corresponding task amount. We will analytically characterize such a relationship later on.}} Without loss of generality, we {order user types in the descending order of the \emph{unit cost}, i.e.,}
$K_1>K_2>...>K_I$, {i.e., a higher type of user has a smaller cost.} {Note that if any two types have the same unit cost, we can group them together as a single type.} {It should also be mentioned that this {unit-cost ordering} is different from the way that we order smartphone users' \emph{constant} collaboration costs in Section~\ref{sec:Data}.}

By asking each type-$i$ user to accomplish the amount $t_i$ of work
and rewarding him with $r_i$, the master's profit is\footnote{{Here the summation operation is taken outside the logarithmic terms, as it is hard to coordinate different user types (in operation systems, computing speeds, and storage capabilities) in the same computing task due to compatibility and synchronization issues and they will be responsible for different tasks at the same time.}}
\begin{equation} \label{eq:pi_ri_ti}
\pi(\{(r_i,t_i)\}_{i\in\mathcal{I}})=\sum_{i
\in\mathcal{I}}\left(\theta_i \log(1+N_i t_i)-N_i r_i\right).
\end{equation}
The term $\theta_i \log(1+N_i t_i)$ {is increasing in users' efforts and} well characterizes the master's
diminishing return (or utility) from the {total work  $N_i t_i$} finished by type-$i$ (as in
\cite{li2009revenue,basar2002revenue}).\footnote{{The assumed logarithmic utility term helps us derive closed-form solutions and engineering insights. Using other concave terms with diminished return are not likely to change the main conclusions.}} The parameter $\theta_i >
0$ characterizes the master's preference for work performed by
{type-$i$ users,} and {does not depend on} $K_i$.\footnote{{The value of $\theta_i$ depends on how easily the master can combine the computing efforts of this type with other types' efforts to finalize the computing result. Notice that $\theta_i$ is large if the master's communication with type-$i$ users is efficient (e.g., with small delay) or the type-$i$'s operation system becomes highly compatible with the master's system to handle computing tasks.}} {Notice that we do not require $\theta_i$'s to be monotonically ordered.}
The term $N_{i}r_{i}$ in
(\ref{eq:pi_ri_ti}) is the total reward that the master {offers
to type-$i$ users.} {The summation operation in (\ref{eq:pi_ri_ti}) is motivated by the fact that many complex engineering or commercial problems can be separated into multiple subproblems and solved in a distributed manner (e.g., \cite{litke2004mobile,palmer2009ibis}).}

By examining (\ref{eq:u_i}) and (\ref{eq:pi_ri_ti}), we can see that
the master and {users} have conflicting objectives. The master wants users to {accomplish a larger
task}, which increases  the master's utility {as well as} users'
collaboration costs. Users want to obtain a larger reward, which
decreases the master's profit. {Next we study how
master and users interact through a contract.}

\subsection{{Master-Users} Contractual Interactions\label{subsec:contract}}

Contract theory studies how an economic decision-maker constructs
contractual arrangements, especially in the presence of asymmetric
(private) information \cite{bolton2005contract}. In our case, the
user types {are}  private information.

The master proposes a contract that specifies the relationship
between a {user's} amount of  {task} $t$ and
reward $r$. Specifically, a {\em contract} is a set $\mathcal{C} =
\{(t_1, r_1), \ldots, (t_M, r_M)\}$ of $M \geq 1$ (amount of
{task}, reward)-pairs that are called {\em contract
items}. The master proposes $\mathcal{C}$. Each user selects a
contract item $(t_m, r_m)$ and performs the amount of work $t_m$ for
the reward $r_m$. According to \cite{bolton2005contract}, it is
optimal for the master to design a contract item for each type,
i.e., $M=I$. {Note that a user can always choose not to work
{for} the master, which implies an implicit contract
item $(r,t)=(0,0)$ (often not counted in the total number of
contract items).} {Once a user accepts some contract item, he
needs to accomplish the task and the master needs to reward him
according to that item.}

{Each type of users} selects the contract item that maximizes his
payoff in (\ref{eq:u_i}). The master wants to optimize the contract
items and maximize his profit in (\ref{eq:pi_ri_ti}).
We {will again focus on a two-stage Stackelberg game,
where the master proposes the contract first and users choose the
contract {items} afterwards.}

Next, we study {how the master} determines the contract that maximizes his profit, depending on what information he has about the
users' types. {As explained in the beginning of Section~\ref{subsec:network_comput}, we assume that a user knows his unit cost. This means that we only need to consider two information scenarios,  complete information and asymmetrically
incomplete information, depending on what the master knows.}

\subsection{Contract Design under Complete Information}\label{subsec:contract_symm}

In this subsection, we study the case where the master knows
{the type of each user {though this case is not easy to realize in practice}}. {The analysis of this subsection mainly serves as a benchmark for understanding the more realistic incomplete information scenario in next subsection.} {Under complete information, it is feasible for the
master to} monitor and make sure that each type of users accepts
only the contract item designed for that type. The master needs to
ensure that each user has a non-negative payoff so that {the user will }accept the contract. In other words, the contract should
satisfy the following {individual rationality}
constraint{s}.

\begin{definition}[IR: Individual Rationality]\label{def:IR}
 A contract satisfies the individual
rationality constraint{s} if each type{-$i$} user receives
a non-negative payoff by accepting the contract item for
type{-$i$}, i.e.,\footnote{{We can easily extend our model by considering a reservation payoff $u_0>0$ for all users, which represents their benefit by making an alternative choice besides joining in the collaboration. By following a similar analysis, we can still show that the new IR constraints $r_i-K_it_i\geq u_0$ are tight at the contract optimality for any type-$i$ user who joins in the collaboration. The key difference is that the master now needs to match users' reservation payoffs by announcing larger rewards, which is slightly different from Theorem 8.}}
\begin{equation}\label{eq:IR}
r_i-K_it_i\geq 0,\ \forall i\in\mathcal{I}.
\end{equation}
\end{definition}

Under complete information, the optimal contract
{$\mathcal{C}=\{(r_i^*,t_i^*)\}_{i\in\mathcal{I}}$} solves the
following problem:
\begin{align}\label{eq:opt_symm}
&\max_{\{(r_i,t_i)\}_{i\in\mathcal{I}}}\pi({\{(r_i,t_i)\}_{i\in\mathcal{I}}})=\sum_{i\in\mathcal{I}}{(\theta_i\log(1+N_it_i)-N_ir_i)},\nonumber
\\ &\text{subject\ to:\  \text{IR\ constraints}\ (\ref{eq:IR})} \text{ and }  0\leq t_i\leq \bar{t}_i,\forall
i\in\mathcal{I}.
\end{align}

It is easy to check that the IR constraints {are tight} at the
{optimal solution} to Problem (\ref{eq:opt_symm}), and the
master {will} leave a zero payoff to each {type-$i$ user
}with $r_i^*=K_it_i^*$. Also, due to the independence of each type
in Problem (\ref{eq:opt_symm}), we can decompose Problem
(\ref{eq:opt_symm}) into $I$ subproblems. For each type
$i\in\mathcal{I}$, {the master needs to solve} the following subproblem
\begin{align}\label{eq:opt_symm_i}
&\max_{t_i}\pi_i(t_i)=\theta_i\log(1+N_it_i)-N_iK_it_i,\nonumber\\
&\text{subject\ to:}\ 0\leq t_i\leq \bar{t}_i.
\end{align}
By solving all $I$ subproblems, we have the following
result.

\begin{theorem}[Optimal Contract under Complete Information]\label{thm:completeinfo}
At the equilibrium, the master will hire {the type-$i$} users {if and only if}
$\theta_i>K_i$. {The total} involved user type set is
\begin{equation}\label{eq:collaborate_set_symm}
\mathcal{I}_C =\{i\in\mathcal{I}: \theta_i>K_i\}.
\end{equation}
The subscript $C$ in $\mathcal{I}_C$ refers to the complete
information assumption.
For a user with type $i\in\mathcal{I}_C$, the
{equilibrium} contract item is
\begin{multline}\label{eq:contract_symm}
(r_i^*,t_i^*)=(K_it_i^*,t_i^*)\\=\left(\min\left(\frac{\theta_i-K_i}{N_i},K_i\bar{t}_i\right),\min\left(\frac{\theta_i-K_i}{K_iN_i},\bar{t}_i\right)\right).
\end{multline}
{For a user with type $i\notin \mathcal{I}_C$, the equilibrium contract item is $(r_i^*,t_i^*)=(0,0)$.}
{All users} (no matter joining collaboration or not)
receive a zero payoff. {The master's} equilibrium profit is
\begin{align}\label{eq:client_symm}
&\pi^*=\sum_{i\in\mathcal{I}_C} \nonumber \\
&\min\left(\theta_i\log\left(\frac{\theta_i}{K_i}\right)-\theta_i+K_i,\theta_i\log(1+N_i\bar{t}_i)-N_iK_i\bar{t}_i\right).
\end{align}
\end{theorem}
The proof of Theorem~\ref{thm:completeinfo} is given in Appendix~{G}.
%

Intuitively, {the master needs to compensate a
collaborator's cost, }thus he will hire type-$i$ users only when his
preference characteristic $\theta_{i}$ is larger than the unit cost of that type
$K_i$.
Users will receive a zero payoff since their private information
about unit costs are known to the master.

By looking into all parameters in the equilibrium contract in
(\ref{eq:contract_symm}) and payoff $\pi^*$ in
(\ref{eq:client_symm}), we have the following observation.

\begin{observation}\label{ob:contract_symm}
For $i\in\mathcal{I}_C$, the equilibrium task $t_i^*$ to a type-$i$
user increases in $\theta_i$, and decreases in $N_i$ and
$K_i$. The master may or may not {offer} a larger task or reward to a
higher type-$i$ collaborator, depending on $N_i$ and
$\theta_i$ for that type. Also, the master's equilibrium profit $\pi^*$  increases
in $\theta_i$, $N_i$, and $\bar{t}_i$, and decreases in $K_i$.
\end{observation}


%



Notice that a higher type-$i$ collaborator has less unit cost where
the master needs to compensate, but the master may not give him a
larger task or reward. {This can happen when there are too many
collaborators of that type, or the master evaluates this type {with a
small value of $\theta_{i}$.}}

\subsection{Master's Contract Design under Asymmetrically Incomplete
Information}\label{subsec:contract_asymm}

In this subsection, we study the case where the master only has
asymmetrically incomplete information about each user's type. A
user's actual type is only known to himself, and the master and the
other users only have a rough estimation on this. We consider that
others believe a user belonging to type-$i$ {with a probability $q_i$}.
Everyone knows the total number of users $N$.\footnote{{Users can know $N$ by checking the master's or some third party's market survey, or the news on recent penetration or shipment of smartphones.}}

\subsubsection{Feasibility of contract under asymmetrically incomplete information}

According to \cite{bolton2005contract}, the master's
contract should first be feasible {in this scenario}. A
feasible contract must satisfy both individual rationality (IR)
constraints (Definition~\ref{def:IR}  in
Section~\ref{subsec:contract_symm}) and incentive compatibility
constraints defined as follows.

\begin{definition}[IC: Incentive Compatibility]\label{def:IC}
A contract satisfies the incentive
compatibility constraints if each type-$i$ user prefers to choose
the contract item for his own type, i.e.,
\begin{equation}\label{eq:IC}
r_i-K_it_i\geq r_j-K_it_j,\;\; \forall i,j\in\mathcal{I}.
\end{equation}
\end{definition}

{{Under asymmetrically incomplete information}, the master does not know the number of users
$N_{i}$ of type-$i$.
Let us denote the} users' numbers of all types as
{$\{n_i\}_{i\in\mathcal{I}}$}, which are random variables {following}
certain distributions and {satisfying}
{$\sum_{i\in\mathcal{I}}n_i=N$}. {Note that the realizations of $\{n_i\}_{i\in\mathcal{I}}$ depend on $N$ and probabilities $\{q_i\}_{i\in\mathcal{I}}$ of all types that a user may belong to.} The master's profit {for a
particular realization of $\{n_i\}_{i\in\mathcal{I}}$} is
\begin{equation}\label{eq:Epi}
\pi(\{(r_i,t_i)\}_{i\in\mathcal{I}},\{n_i\}_{i\in\mathcal{I}})=\sum_{i\in\mathcal{I}}{(\theta_i\log(1+n_it_i)-n_ir_i)}.
\end{equation}
Thus the master's
expected profit is
\begin{align}\label{eq:Expected_pi}
& \mathbb{E}_{\{n_i\}_{i\in\mathcal{I}}}[\pi(\{(r_i,t_i)\}_{i\in\mathcal{I}},\{n_i\}_{i\in\mathcal{I}})]\nonumber\\&=\sum_{n_1=0}^N\sum_{n_2=0}^{N-n_1}...\sum_{n_{I-1}=0}^{N-\sum_{j=1}^{I-2}n_j}\frac{N!q_1^{n_1}...q_{I-1}^{n_{I-1}}q_I^{N-\sum_{j=1}^{I-1}n_j}}{n_1!...n_{I-1}!(N-\sum_{j=1}^{I-1}n_j)!}\nonumber\\
& \ \ \ \ \cdot \pi(\{(r_i,t_i)\}_{i\in\mathcal{I}},\{n_i\}_{i\in\mathcal{I}}).
\end{align}
{The master's profit optimization problem as\footnote{{{After observing all $I$ reward-task combinations in $\{(r_i,t_i)\}_{i\in\mathcal{I}}$, it is optimal for any type-$i$ user to choose only $(r_i, t_i)$ of his own type later, according to IR and IC constraints in the contract mechanism.}}}}
\begin{align}\label{eq:opt_asymm}
&\max_{\{(r_i,t_i)\}_{i\in\mathcal{I}}}\mathbb{E}_{\{n_i\}_{i\in\mathcal{I}}}[\pi(\{(r_i,t_i)\}_{i\in\mathcal{I}},\{n_i\}_{i\in\mathcal{I}})]\nonumber
\\ &\text{subject\ to:}\  \text{IR\ constraints\ in\ (\ref{eq:IR})},\nonumber\\
&\ \ \ \ \ \ \ \ \ \ \ \ \ \text{IC\ constraints\ in\ (\ref{eq:IC})},\nonumber\\
&\ \ \ \ \ \ \ \ \ \ \ \ \ 0\leq t_i\leq \bar{t}_i, \forall
i\in\mathcal{I}.
\end{align}
{The total number of IR and IC constraints is $I^{2}$.}  {Next, we show that it is possible to represent these $I^{2}$ constraints with a set of much fewer equivalent constraints. }

\begin{proposition}\label{prop:feasible}
{\emph{(Sufficient and Necessary Conditions for feasibility):}} For a contract
$\mathcal{C}=\{(r_i,t_i),\forall i\in\mathcal{I}\}$ {with} {user
costs} $K_1>...>K_I$, it is feasible if and only if all the
following conditions are satisfied:
\begin{enumerate}
\item $\mathtt{Condition}  (+)$: $r_1-K_1t_1\geq 0$;
\item $\mathtt{Condition}  (\uparrow)$: $0\leq r_1\leq ...\leq r_I$ and
$0\leq t_1\leq ...\leq t_K$;
\item
$\mathtt{Condition}  (\leq)$: For any $i=2,...,I$,
\begin{equation}\label{eq:contract_leq} r_{i-1}+K_i(t_i-t_{i-1})\leq
r_i\leq r_{i-1}+K_{i-1}(t_i-t_{i-1}).
\end{equation}
\end{enumerate}
\end{proposition}

The proof of Proposition~\ref{prop:feasible} is given in Appendix~{H}.

Intuitively, $\mathtt{Condition}  (+)$ {ensures
that all types of users can get a nonnegative payoff by accepting
the contract item $(r_{1},t_{1})$, as it implies
$r_{1}-K_{j}t_{1}\geq 0$ for all $j\geq 2$. Thus this can replace
the IR constraints in (\ref{eq:IR}).}
$\mathtt{Condition}  (\uparrow)$ and $\mathtt{Condition}  (\leq)$
are related to IC constraints in (\ref{eq:IC}). $\mathtt{Condition}
(\uparrow)$ shows that a user with a higher type should be assigned
a larger task, because his unit cost is lower (and more efficient)
and the master needs to compensate this user less {per unit
work}. Also, a larger reward should be given to this user for the
larger task undertaken by him, {otherwise this user will choose
another contract item in order to work less}. $\mathtt{Condition}
(\leq)$ shows the relation between any two {neighboring contract
items.}

Based on Proposition~\ref{prop:feasible}, we can simplify the
master's problem {in (\ref{eq:opt_asymm})} as
\begin{align}\label{eq:opt_asymm_sim}
&\max_{\{(r_i,t_i)\}_{i\in\mathcal{I}}}\mathbb{E}_{\{n_i\}_{i\in\mathcal{I}}}[\pi(\{(r_i,t_i)\}_{i\in\mathcal{I}},\{n_i\}_{i\in\mathcal{I}})]\nonumber
\\ &\text{subject\ to,}\  \mathtt{Condition}  (+),\mathtt{Condition}  (\uparrow),\mathtt{Condition}  (\leq),\nonumber\\
&\ \ \ \ \ \ \ \ \ \ \ \ \ 0\leq t_i\leq \bar{t}_i, \forall
i\in\mathcal{I},
\end{align}
where the previous $I^2$  IR and IC constraints have been reduced to $I+2$ constraints.

\subsubsection{Analysis by sequential optimization}

Now we want to solve the master's optimal contract. However,
(\ref{eq:opt_asymm_sim}) is not easy to solve as it has
{coupled} variables and
{many} constraints.
{The way we solve is} a sequential optimization approach: we
first derive the optimal rewards
$\{r_i^*(\{t_i\}_{i\in\mathcal{I}})\}_{i\in\mathcal{I}}$ given any
feasible tasks $\{t_i\}_{i\in\mathcal{I}}$, then  further derive
the optimal tasks $\{t_i^*\}_{i\in\mathcal{I}}$ for the optimal
contract.

\begin{proposition}\label{prop:opt_r_t}
Let $\mathcal{C}=\{(r_i,t_i)\}_{i\in\mathcal{I}}$ be a feasible
contract with any feasible tasks $ 0\leq t_1\leq ...\leq t_I$. The unique optimal rewards {$\{r_i^*(\{t_i\}_{i\in\mathcal{I}})\}_{i\in\mathcal{I}}$} satisfy
\begin{align}\label{eq:opt_r_t1}
r_1^*\left(\{t_i\}_{i\in\mathcal{I}}\right)&=K_1t_1,\\
r_i^*\left(\{t_i\}_{i\in\mathcal{I}}\right)&=r_{i-1}^*+K_i(t_i-t_{i-1})\nonumber\\
&\!\!\!\!\!\!\!\!\!\!\!\!=K_1t_1+\sum_{j=2}^iK_j(t_j-t_{j-1}), \forall
i=2,...,I,\label{eq:opt_r_ti}
\end{align}
{Notice that the lowest type user obtains a zero payoff, and a user's optimal payoff is non-decreasing in his type.}
\end{proposition}

\begin{proof}[Proof (Sketch)] First, we can prove (\ref{eq:opt_r_t1}) by
showing that $\mathtt{Condition}  (+)$ binds at the optimality. This
 guarantees the IR constraints of the contract.
Second, we can prove (\ref{eq:opt_r_ti}) by showing that the
left-hand side inequality in $\mathtt{Condition}  (\leq)$ binds at the
optimality. This guarantees the IC constraints of the
contract. {Finally, (\ref{eq:opt_r_t1}) shows a zero payoff for the lowest type user, and (\ref{eq:opt_r_ti}) shows that for any $\{t_i\}_{i=1}^I$,
\[
r_i^*-K_it_i=r_{i-1}^*-K_it_{i-1},
\]
which is no smaller than $r_{i-1}^*-K_{i-1}t_{i-1}$ due to $K_i<K_{i-1}$. Thus a user's payoff is non-decreasing in his type.}
\end{proof}

Based on Proposition~\ref{prop:opt_r_t}, we can greatly simplify the
master's optimization Problem in (\ref{eq:opt_asymm_sim}) as
\begin{align}\label{eq:opt_asymm_final}
&\max_{\{t_i\}_{i\in\mathcal{I}}}\mathbb{E}_{\{n_i\}_{i\in\mathcal{I}}}\left[\pi(\{(r_i^*(\{t_i\}_{i\in\mathcal{I}}),t_i)\},\{n_i\}_{i\in\mathcal{I}})\right]\nonumber\\
&\text{subject\ to,}\ 0\leq t_1\leq ...\leq t_I,\nonumber\\
& \ \ \ \ \ \ \ \ \ \ \ \ \ t_i\leq\bar{t}_i, \forall
i\in\mathcal{I}.
\end{align}

Next we first examine whether Problem (\ref{eq:opt_asymm_final}) is a convex problem and then derive a way to solve it. The
first derivative of the master's expected profit over $t_i$ is
\begin{align}\label{eq:der_E_pi}
&\frac{\partial \mathbb{E}_{\{n_i\}_{i\in\mathcal{I}}}[\pi(\{(r_i^*(\{t_i\}_{i\in\mathcal{I}}),t_i)\},\{n_i\}_{i\in\mathcal{I}})]}{\partial t_i}\nonumber\\
&=\sum_{n_1=0}^N\sum_{n_2=0}^{N-n_1}...\sum_{n_{I-1}=0}^{N-\sum_{j=1}^{I-2}n_j}\frac{N!q_1^{n_1}...q_{I-1}^{n_{I-1}}q_I^{N-\sum_{j=1}^{I-1}n_j}}{n_1!...n_{I-1}!(N-\sum_{j=1}^{I-1}n_j)!}\nonumber\\
& \ \ \ \
\left(\frac{n_i\theta_i}{1+n_it_i}-n_iK_i-(K_i-K_{i+1})\sum_{\forall
j>i, \forall j\in\mathcal{I}}n_j\right),\forall i\in\mathcal{I},
\end{align}
which is independent of $t_j$ for any $j\neq i$. From (\ref{eq:der_E_pi}), we can easily check the Hessian matrix of the objective function in Problem (\ref{eq:opt_asymm_final}) and conclude that it is concave in $\{t_i,i\in\mathcal{I}\}$. Furthermore, the constraints in Problem (\ref{eq:opt_asymm_final}) are all linear and the feasible set in Problem (\ref{eq:opt_asymm_final}) is convex {and not empty}. Thus we conclude that Problem (\ref{eq:opt_asymm_final}) is convex. It should be noted that the feasible set of Problem (\ref{eq:opt_asymm_final}) has interior such that the strick inequalities hold for all constraints in Problem (\ref{eq:opt_asymm_final}). For example, for any positive value $\eta$, such an interior point {feasible for Problem (\ref{eq:opt_asymm_final})} could be $$t_i=\frac{\min(\{\bar{t}_i,\forall i\in\mathcal{I}\})}{I-i+1+\eta}, \forall i\in\mathcal{I}.$$ Thus Problem (\ref{eq:opt_asymm_final}) is a convex problem satisfying Slater's condition (implying strong duality) and always has a solution, and can be optimally solved by examining KKT conditions.

The Lagrangian function is \begin{align}&\quad L(\{t_i,i\in\mathcal{I}\}, \{\lambda_i,i\in\mathcal{I}\setminus \{I\}\},\{v_i,i\in\mathcal{I}\})\nonumber\\ & =\mathbb{E}_{\{n_i\}_{i\in\mathcal{I}}}\left[\pi(\{(r_i^*(\{t_i\}_{i\in\mathcal{I}}),t_i)\},\{n_i\}_{i\in\mathcal{I}})\right]+\sum_{i\in\mathcal{I}}v_i(\bar{t}_i-t_i)\nonumber\\
&\quad +\sum_{i\in\mathcal{I}\setminus \{I\}}\lambda_i(t_{i+1}-t_i),\end{align}
where $\{\lambda_i,i\in\mathcal{I}\setminus \{I\}\}$ and $\{v_i,i\in\mathcal{I}\}$ are Lagrange multipliers corresponding to the constraints in Problem (\ref{eq:opt_asymm_final}). The KKT conditions are as follows.
\begin{itemize}
\item \emph{Primal constraints:} $t_i^*\leq t_{i+1}^*, \forall i\in\mathcal{I}\setminus \{I\}$; $t_i^*\leq \bar{t}_i,\forall i\in\mathcal{I}$;
\item \emph{Dual constraints:} $\lambda_i^*\geq 0, \forall i\in\mathcal{I}\setminus \{I\}$, and $v_i^*\geq 0, \forall i\in\mathcal{I}$;
\item \emph{Complementary slackness:} $\lambda_i^*(t_{i+1}^*-t_i^*)=0, \forall i\in\mathcal{I}\setminus \{I\}$, and $v_i^*(\bar{t}_i-t_i^*)=0, \forall i\in\mathcal{I}$;
\item \emph{First-order condition of Lagrangian with respect to $t_i$:}
\begin{align}\label{eq:firstorder}
{\partial L}/{\partial t_1}=&{\partial \mathbb{E}_{\{n_i\}_{i\in\mathcal{I}}}\left[\pi(\{(r_i^*(\{t_i\}_{i\in\mathcal{I}}),t_i)\},\{n_i\}_{i\in\mathcal{I}})\right]}/{\partial t_1}\nonumber\\&-\lambda_1-v_1=0,\nonumber\\
{\partial L}/{\partial t_i}=&{\partial \mathbb{E}_{\{n_i\}_{i\in\mathcal{I}}}\left[\pi(\{(r_i^*(\{t_i\}_{i\in\mathcal{I}}),t_i)\},\{n_i\}_{i\in\mathcal{I}})\right]}/{\partial t_i}\nonumber\\&-(\lambda_i-\lambda_{i-1})-v_i=0,\forall i\in\mathcal{I}\setminus \{1,I\},\nonumber\\
{\partial L}/{\partial t_I}=&{\partial \mathbb{E}_{\{n_i\}_{i\in\mathcal{I}}}\left[\pi(\{(r_i^*(\{t_i\}_{i\in\mathcal{I}}),t_i)\},\{n_i\}_{i\in\mathcal{I}})\right]}/{\partial t_I}\nonumber\\&+\lambda_{I-1}-v_I=0,\end{align}
from which we cannot derive closed-form solutions but can rely on numerical methods (e.g., primal dual algorithm) to show numerical results later on.\footnote{It should be mentioned that some multiplier $\lambda_i$ (corresponding to the constraint $t_{i+1}\geq t_i$) may be nonzero when the master has much smaller preference characteristics on higher user type-$(i+1)$ than type-$i$, or the higher type involves many more users than the lower type.
Some multiplier $v_i$ (corresponding to $t_i\leq \bar{t}_i$) may be nonzero when the capacity upper bound $\bar{t}_i$ of type-$i$ is small.}
{The computation complexity to solve Problem (\ref{eq:opt_asymm_final}) is not high and the complexity upperbound can be derived in the following way. Due to the task relationships among different user types (i.e., $t_1\leq ...\leq t_I$ and $t_I\leq \bar{t}_I$), the possible range of each task $t_i$ is $[0,\bar{t}_I]$. We can approximate the continuity of this range through a proper discretization, i.e., representing all possibilities of any $t_i$ by $T$ equally spaced values (with the first and last values equal to $0$ and $\bar{t}_I$, respectively). By {(approximately) solving} all the KKT conditions especially (\ref{eq:firstorder}) above, we require computation in order $\mathcal{O}(T)$ to search over all $T$ possibilities for each optimal $t_i^*$ for type-$i$. The overall computation complexity for all $I$ types is $\mathcal{O}(I\cdot T)$ in Problem (\ref{eq:opt_asymm_final}). {The choice of $T$ will affect the quantization error of the computation.}
}
\end{itemize}


Actually, {without} explicitly solving Problem~(\ref{eq:opt_asymm_final}), we can
still derive some interesting results by looking into the KKT conditions as follows.
\begin{theorem}\label{thm:incompleteinfo}
{The total involved user type set} under asymmetrically incomplete information is
\begin{align}\label{eq:collaborate_set_asymm}
&\!\!\!\mathcal{I}_{A}=\{i\in\mathcal{I}: \nonumber\\
&\!\!\!\mathbb{E}_{\{n_i\}_{i\in\mathcal{I}}}[n_i(\theta_i-K_i)-(K_i-K_{i+1})\sum_{\forall
j>i, j\in \mathcal{I}}n_j]>0\},
\end{align}
where the subscript $A$ in $\mathcal{I}_A$ refers to the asymmetrically
incomplete information assumption.\footnote{{Note that the
master will design $(r^*,t^*)=(0,0)$ for the types not in set
$\mathcal{I}_A$. Thus the users of these types are not involved as
collaborators.}} Compared {with} the {collaborator set
$\mathcal{I}_{C}$} under complete information case, {here the
master involves less collaborators, i.e., $|\mathcal{I}_{A}|\leq
|\mathcal{I}_{C}|$.} {Moreover, the master assigns a larger
task and gives a larger reward to  a higher type of collaborator,
which may not be the case under complete information (see
Observation~\ref{ob:contract_symm})}. {Only the lowest
type of collaborator(s) in set $\mathcal{I}_{A}$ obtains a zero
payoff, and higher types of collaborators in set $\mathcal{I}_{A}$
obtain positive payoffs that are {increasing} in their types.}
\end{theorem}
The proof of Theorem~\ref{thm:incompleteinfo} is given in Appendix~{I}.

\begin{figure}[tt]
\centering
\includegraphics[width=0.45\textwidth]{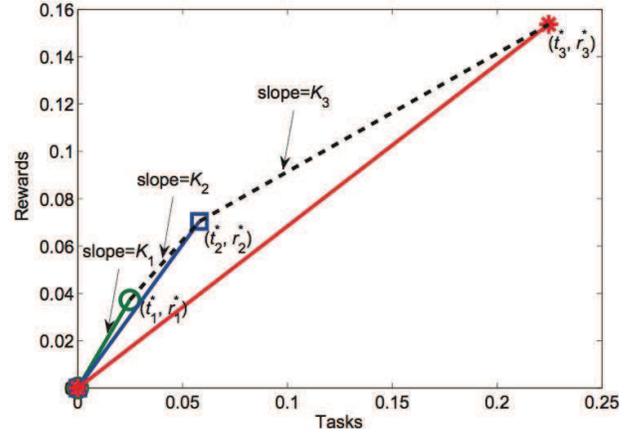}
\caption{The master's {optimal} contract items for three types
($I$=3). Other parameters are $N=120$, $K_1=1.5$, $K_2=1$,
$K_3=0.5$, $\theta_i=5$, and $q_i=1/3$ for any $i\in\mathcal{I}$.}
\label{fig:contract_t_r}
\end{figure}

Intuitively, as the master {does not know each user's type, he
needs to provide incentives (in terms of positive payoffs) to the
users to attract them revealing their own types truthfully.}
{If he involves a low type user, he needs to give increasingly
higher payoffs to all higher types. Thus he should target at users with high enough types.}  We have $|\mathcal{I}_{A}|$ smaller than
$|\mathcal{I}_{C}|$, {which means that some low types belong to
set $\mathcal{I}_{C}$ {may not be} included in set
$\mathcal{I}_{A}$.} {By comparing
(\ref{eq:collaborate_set_asymm}) and (\ref{eq:collaborate_set_symm})
for the highest type-$I$, we know that that this type is involved in
both information scenarios.}

Recall that under complete information,
Observation~\ref{ob:contract_symm} shows that the master may
not give {a} larger task and reward to a higher type-$i$
collaborator. This can happen when $\theta_i$ is small or the number
of users of that type is large. {Under asymmetrically incomplete
information, however, }the IC constraints require {the reward
and task to be nondecreasing in the collaborator types, independent
of $\theta_{i}$ and the number of users in each type (which is a
random variable).} {Otherwise, some collaborators will have
incentives to choose contract items not designed for their own
types, and thus violate IC constraints. This is not optimal for the master based on the}
Revelation Principle\cite{bolton2005contract}.

\begin{figure}[t]
\centering
\includegraphics[width=0.42\textwidth]{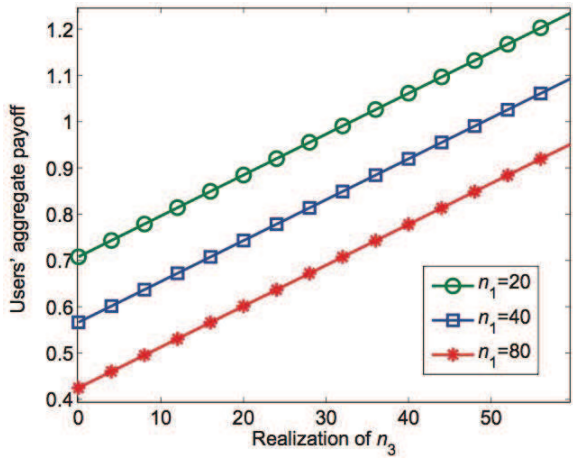}
\caption{Users' aggregate payoff under asymmetrically incomplete information as
a function of users' realized numbers $\{n_i\}_{i=1}^3$ in three
types ($I$=3). {Here we only show $n_{1}$ and $n_{3}$, and
$n_{2}$ can be computed as $N-n_{1}-n_{3}$.} Other parameters are
$N=120$, $K_1=1.1$, $K_2=1$, $K_3=0.9$, $\theta_i=5$, and $q_i=1/3$ for
any $i\in\mathcal{I}$.
}
\label{fig:aggregateusers}
\end{figure}

Figure~\ref{fig:contract_t_r} shows the master's optimal contract
$\{(r_i^*,t_i^*)\}_{i=1}^3$ for three collaborator types. {A
higher type-$i$ user obtains a larger task $t_{i}^{\ast}$, a larger
reward $r_{i}^{\ast}$, and a larger payoff ($u_i^*=r_i^*-t_i^*$). {This is consistent with Proposition~\ref{prop:opt_r_t}.}}
The slope of {the} dashed line between two points
$(r_i^*,t_i^*)$ and $(r_{i+1}^*,t_{i+1}^*)$ equals to cost $K_{i+1}$
(as shown in Proposition~\ref{prop:opt_r_t}). In the contract, the
ratio between the reward and task (i.e., $r_i^*/t_i^*$) for type-$i$
decreases with the type. Thus a lower type $j<i$ collaborator will
not choose {the higher contract item} ($r_i^*,t_i^*$), since it
is too costly and not {be} efficient for him to undertake the task. A
user will not choose a lower type contract item either, otherwise
his payoff (though still positive) will decrease with a smaller
reward.

By looking into (\ref{eq:der_E_pi}), we have the following result.
%

\begin{observation}\label{ob:ti_thetai}
The master's optimal task allocation $t_i^*$ to a {type-$i$
collaborator increases in the master's preference
{characteristic} $\theta_i$ and decreases in the collaborator's
cost $K_i$. The master's equilibrium expected profit increases in
$\theta_i$ for all {$i\in\mathcal{I}_A$}.}
\end{observation}

{Given the task-reward combinations in the contract, users will benefit from keeping their private information from the master: the lowest type collaborator obtains a zero payoff and a higher type one obtains a larger and positive payoff as in Proposition~\ref{prop:opt_r_t} and Fig.~\ref{fig:contract_t_r}. To understand how the hidden information benefits the entire user population, Fig.~\ref{fig:aggregateusers} investigates users' aggregate payoff as we vary the number of users of each of the three types, $\{n_i\}_{i=1}^3$. The total population is fixed at a size of $N=120$. We can see that the users' aggregate payoff decreases as we have more low type users ($n_1$), and increases as we have more high type users ($n_{3}$). Intuitively, a higher number of type-1 collaborators ($n_1$) means that more collaborators receive a zero payoff, while a higher number of type-3 collaborators ($n_3$) means that more collaborators receive the maximum payoff.}

%



\begin{figure}[tt]
\centering
\includegraphics[width=0.45\textwidth]{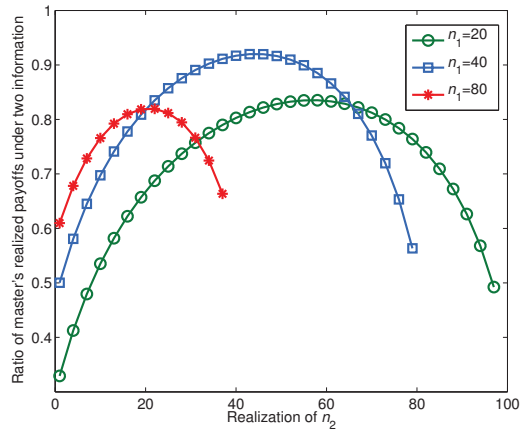}
\caption{The ratio of the master's {realized
payoffs} under asymmetrically incomplete and complete information as
a function of users' realized numbers $\{n_i\}_{i=1}^3$ in three
types ($I$=3). {Here we only show $n_{1}$ and $n_{2}$, and
$n_{3}$ equals $N-n_{1}-n_{2}$.} Other parameters are
$N=120$, $K_1=1.1$, $K_2=1$, $K_3=0.9$, $\theta_i=5$, and $q_i=1/3$ for
any $i\in\mathcal{I}$.
}
\label{fig:contract_asymm_symm_ratio}
\end{figure}

{It should be noted that the master prefers a large probability of having high type users, as these users are more efficient in performing computing tasks given the same reward. As an example, consider three types of users ($I=3$): the master enjoys the maximum collaboration benefit when all users belong to the highest type (i.e., $q_1=q_2=0$ and $q_3=1$). This is also illustrated in Fig.~\ref{fig:contract_t_r}, where the ratio between task and reward (i.e., $t_i^*/r_i^*$) is the highest for the type-3 users. When all users always belong to the same type, users cannot hide their type information from the master, and the master can hire them by just provding a zero payoff. However, when users have positive probabilities of belonging to different types, they can hide their type information from the master, and the master needs to provide more rewards to motivate high type users to contribute.}

Next, we compare the master's profits under complete and
asymmetrically incomplete information.

\begin{observation}
{Compared {with} complete information, the master obtains a
smaller equilibrium expected profit under asymmetrically incomplete
information. The gap between his realized profit under two information scenarios
is minimized when the realization (users' {numbers}
in all types)  is the
closest to the expected value.}
\end{observation}


Figure~\ref{fig:contract_asymm_symm_ratio} shows the ratio of the
master's realized payoffs under asymmetrically incomplete and
complete information, which is a function of users' realizations
$\{n_i\}_{i=1}^3$ in all three types. {This ratio is always no
larger than 1, as the master obtains the maximum profit under
complete information.} {This profit ratio reaches  its maximum
92\% when users' type realization matches the expected value, i.e.,
$n_i=Nq_i=40$ for $i=1,2$ (and thus $n_{3}=N-n_{1}-n_{2}=40$ as
well).} This is {consistent with the fact that the master maximize his expected profit} under asymmetrically incomplete information. Note that even in this case, there is still a profit loss for the master under asymmetrically incomplete information due to the loss of information.

\section{Conclusion}

This paper analyzes different mechanisms that a master
can use to motivate the collaboration of smartphone users on both
\emph{data acquisition} and \emph{distributed computing}. {Our
proposed incentive mechanisms cover several possible information
{scenarios}  that the master may face in
reality.}
For data acquisition applications, we propose a reward-based
collaboration scheme {for the master to attract
{enough} users by giving out {the minimum} reward}.
For distributed computing applications, we use contract theory to
study how a master decides different task-reward combinations
for {many} different types of users.

There are some possible ways to extend the results in this paper.
For the data acquisition applications, for example, we can consider a flexible
revenue model instead of a threshold one. {For example, Google
can still benefit if a few users take pictures of some critical events.}
{The master will still give out some reward even facing a small number of users, and his reward and profit would increase as more and more users choose to collaborate. } {Moreover, in some network with small number of users, the geographical positions of users could be more important than the total number. We will study such an spatial issue in the future. }

\bibliographystyle{IEEEtran}

\appendices

\section{Proof of Theorem~\ref{thm:data_complete}}\label{app:data_complete}
If $V<n_0C_0$, then the master's announced total reward $R$ is also
smaller than $n_0C_0$ to make a profit. This reward is not enough to
compensate even $n_0$ users with smallest costs, thus no users will
join. Regarding this, the master will not seek users' collaboration
in Stage I by announcing zero reward $R^*=0$. Next we focus on
$V\geq n_0C_0$.

We first prove the existence of the equilibrium in
Theorem~\ref{thm:data_complete}. In the strategies shown in
Theorem~\ref{thm:data_complete}, involved users will not leave the
collaboration since they have non-negative payoffs. Also, those
users not in the collaboration will not decide to collaborate,
otherwise they receive negative payoffs. The master will not deviate
by decreasing or increasing the $R^*$, otherwise he will involve
less than $n_0$ users or loss profit, respectively.

We then prove the uniqueness of the equilibrium by contradiction. Note that
$R^*=n_0C_0$ corresponds to a unique state of users' equilibrium decisions in
Theorem~\ref{thm:data_complete}. Suppose there
exists another equilibrium with a different $\hat{R}^*\neq R^*$.
If $\hat{R}^*< R^*$, the master
cannot attract enough collaborators and the collaboration is not successful; if $\hat{R}^*>R^*$, the master
has incentive to decrease $\hat{R}^*$ to $R^*$. Thus there does
exist such an equilibrium with $\hat{R}^*\neq R^*$.

\section{Proof of No Collaboration and Pure Strategy NE in Theorem~\ref{them2}}\label{app:SymmIncomInfo}
We focus on users' pure strategies where $R$ is already given. If
$R<n_0\mu$, this reward cannot attract $n_0$ collaborators where
each user's collaboration cost is believed to be $\mu$. Thus the
collaboration is not successful and no user will collaborate in
Stage II. Next we focus on $R\geq n_0\mu$.
\begin{itemize}
\item If $n_0\mu\leq R<N\mu$, we prove $n^*=\lfloor
\frac{R}{\mu}\rfloor$ by contradiction. Suppose there are $n^*\neq
\lfloor \frac{R}{\mu}\rfloor$ collaborators at the equilibrium.
\begin{itemize}
\item If $n^*<\lfloor \frac{R}{\mu}\rfloor$, then another user will
join the collaboration and receive nonnegative expected payoff
(nonnegative payoff $\frac{R}{n^*+1}-\mu$ when collaboration is
successful and zero payoff otherwise).
\item If $n^*>\lfloor
\frac{R}{\mu}\rfloor$, then some involved user will leave the
collaboration since he receives negative expected payoff (negative
payoff $\frac{R}{n^*}-\mu$ if the collaboration is successful and
zero payoff otherwise).
\end{itemize}
Thus there are $n^*=\lfloor \frac{R}{\mu}\rfloor$ collaborators at
the equilibrium.
\item If $R\geq N\mu$, each user can join the collaboration and
receive non-negative expected payoff and thus $n^*=N$.
\end{itemize}


\section{Proof of Existence And Uniqueness of Equilibrium Threshold in Theorem~\ref{them4}}\label{app:AsymmInfoStageII}
Recall that $\Phi(\gamma)$ is given in (\ref{eq:StageII_asymm}). Here we want to prove that there exists a unique solution $\gamma^*(R)$ (or simply $\gamma^*$) to
$\Phi(\gamma)=0$, which satisfies
$\frac{R}{N}<\gamma^*<\frac{R}{n_0}$.

We divide the proof into the following three parts, depending on
relation between $R$ and $\gamma^*$. For simplicity, we represent
$F(\gamma^*)$ as $F^*$.
\begin{itemize}
\item Suppose that there exists a solution $\gamma^*$ to
$\Phi(\gamma)=0$ in (\ref{eq:StageII_asymm}) which satisfies $R\leq
n_0\gamma^*$. Since $\Phi(\gamma^*)$ is increasing in $R$, we have
$\Phi(\gamma^*)\leq \Phi(\gamma^*)\mid_{R=n_0\gamma^*}$. That is,
\begin{align}\Phi(\gamma^*)\leq & \sum_{m=n_0-1}^{N-1}\left(\frac{n_0\gamma^*
}{m+1}-\gamma^*\right)\left(\begin{array}{c}
N-1\\
m\end{array}\right)\nonumber\\
&\cdot(F^*)^m(1-F^*)^{N-1-m},\nonumber\end{align} which is negative
due to our consideration of $n_0<N$ and $F*>0$. Thus there does not
exist any solution $\gamma^*$ to $\Phi(\gamma)=0$ satisfying $R\leq
n_0\gamma^*$ in Stage II.

\item Suppose that there exists a solution $\gamma^*$ to
$\Phi(\gamma)=0$ which satisfies $R\geq N\gamma^*$. We have
$\Phi(\gamma^*)\geq \Phi(\gamma^*)\mid_{R=N\gamma^*}$. That is,
\begin{align}
\Phi(\gamma^*)\geq &
\sum_{m=n_0-1}^{N-1}\left(\frac{N\gamma^*}{m+1}-\gamma^*\right)\left(\begin{array}{c}
N-1\\
m\end{array}\right)\nonumber\\
&\cdot(F^*)^m(1-F^*)^{N-1-m},\nonumber
\end{align} which is
positive due to our consideration of $n_0<N$ and $F^*>0$. Thus there
does not exist any solution $\gamma^*$ to $\Phi(\gamma)=0$
satisfying $R\geq N\gamma^*$ in Stage II.

\item When $n_0\gamma<R< N\gamma$, we first show that there exists
a solution $\gamma^*$ to $\Phi(\gamma)=0$ and then prove its
uniqueness. We can check that $\lim_{\gamma\rightarrow
(R/N)^+}\Phi(\gamma)>0$ and $\lim_{\gamma\rightarrow
(R/n_0)^-}\Phi(\gamma)<0$.
Due to the continuity of $\Phi(\gamma)$ on $\gamma$, there exists a
solution $\gamma^*$ to $\Phi(\gamma)=0$. Next we prove the
uniqueness of the solution by contradiction.

Suppose there exist at least two different solutions to
$\Phi(\gamma)=0$. The first derivative $\Phi(\gamma)$ over $\gamma$
at one solution (denoted as ${\gamma}^*$ with corresponding $F^*$)
is nonnegative. But we have
\begin{align}
&\frac{\partial \Phi({\gamma^*})}{\partial
\gamma}=\sum_{m=n_0-1}^{N-1}
\left(\begin{array}{c}N-1\\m\end{array}\right){(F^*)}^{m-1}(1-F^*)^{N-m-2}\nonumber\\&\cdot[-F^*(1-F^*)+
\frac{dF^*}{d\gamma}(\frac{R}{m+1}-{\gamma}^*)(m+(1-N){F^*})],\nonumber
\end{align}
which is smaller than
\begin{align}\label{eq:proof_uniq}
\sum_{m=n_0-1}^{N-1}&
\left(\begin{array}{c}N-1\\m\end{array}\right)(F^*)^{m-1}(1-F^*)^{N-m-2}\nonumber\\
&\cdot[\frac{d{F^*}}{d\gamma}(\frac{R}{m+1}-{\gamma^*})(m+(1-N)F^*)].
\end{align}
By substituting $\Phi(\gamma^*)=0$ with $F=F^*$ into
(\ref{eq:proof_uniq}), we can show
\begin{align}
\frac{\partial \Phi({\gamma^*})}{\partial
{\gamma}}<&\sum_{m=n_0-1}^{N-1}\left(\begin{array}{c}
N-1\\m\end{array}\right){(F^*)}^{m-1}(1-F^*)^{N-m-2}\nonumber\\&
\cdot\left(\frac{R}{m+1}-{\gamma}^*\right)m\frac{d{F^*}}{d\gamma}<0,\label{eq:Phi_gamma}
\end{align}
where $F(\cdot)$ is an increasing function. This contradicts with
our supposition that the first derivative of
$\partial\Phi(\gamma^*)/\partial \gamma$ is nonnegative. This ends
our proof of the existence of unique solution $\gamma^*$ to
$\Phi(\gamma)=0$.
\end{itemize}
%

\section{Proof of Theorem~\ref{thm:gamma_R_n0_N}}\label{app:gamma_R_n0_N}
We first prove the relation between $\gamma^*$ and $R$. Recall that
(\ref{eq:Phi_gamma}) has shown that $\Phi(\gamma^*)$ is decreasing
in $\gamma^*$, while (\ref{eq:StageII_asymm}) shows that
$\Phi(\gamma^*)$ is linearly increasing in $R$. By applying implicit
function theorem, we can derive
\[
\frac{d \gamma^*}{d R}=-\frac{\partial \Phi(\gamma^*)}{\partial
\gamma}/\frac{\partial \Phi(\gamma^*)}{\partial R}>0.
\]
Thus $\gamma^*$ is increasing in $R$.

Next we prove the relation between $\gamma^*$ and $n_0$. Let us
denote $F(\gamma^*)$ as $F^*$ and define
$$\phi(m):=\left(\frac{R}{m+1}-\gamma^*\right)\left(\begin{array}{c}
N-1\\
m\end{array}\right)(F^*)^m(1-F^*)^{N-1-m},$$ then we can rewrite
$\Phi(\gamma^*)$ in (\ref{eq:StageII_asymm}) as
$\sum_{n_0-1}^{N-1}\phi(m)$. Since
Section~\ref{app:AsymmInfoStageII} shows that $n_0\gamma^*<R<
N\gamma^*$, $\phi(m)$ is positive when $m$ is small and is negative
when $m$ is large. As $n_0$ increases to $n_0+1$, previous positive
term $\phi(n_0-1)$ in $\Phi(\gamma^*)$ disappears while all negative
terms still remain. Hence, $\Phi(\gamma^*)$ decreases with current
$n_0$. Recall that we have shown in (\ref{eq:Phi_gamma}) that
$\Phi(\gamma^*)$ is decreasing in $\gamma^*$, thus $\gamma^*$ is
decreasing in $n_0$ due to $\Phi(\gamma^*)=0$.

Next we prove the relation between $\gamma^*$ and $N$. As $N$
increases to $N+1$, we have an additional negative term $\phi(N)$
appeared in the (\ref{eq:StageII_asymm}) (denoted by
$\tilde{\Phi}(\gamma^*)$). For a previous term $\phi(m)$ with
$n_0-1\leq m\leq N-1$, it changes to
$$\tilde{\phi}(m)=\left(\frac{R}{m+1}-\gamma^*\right)\left(\begin{array}{c}
N\\
m\end{array}\right)(F^*)^m(1-F^*)^{N-m}.$$ Thus we can rewritten
$\tilde{\phi}(m)=(1-F^*){\phi}(m)\frac{N}{N-m}$, where the fraction
term is increasing in $m$. Then the absolute value of a previously
negative term ${\phi}(m)$ (with large $m$) is relatively enlarged
compared to a positive term (with small $m$). Hence, the summation
of the first $N$ terms in $\tilde{\Phi}(\gamma^*)$ is negative, and
$\tilde{\Phi}(\gamma^*)$ with an additional negative term $\phi(N)$
is further decreased to be negative. Recall that we have shown in
(\ref{eq:Phi_gamma}) that $\Phi(\gamma^*)$ is decreasing in
$\gamma^*$, thus $\gamma^*$ is decreasing in $N$ due to
$\Phi(\gamma^*)=0$.


\section{Proof of Theorem~\ref{thm:data_asymm_f}}\label{app:data_asymm_f}
Recall that (\ref{eq:asymm_f}) shows that $f(R)$ is linearly
increasing in $V$ for any $R$ values, thus the master's equilibrium
expected payoff $f(R^*)$ is increasing in $V$.

Next we prove that $f(R)$ and $f(R^*)$ are decreasing in $n_0$.
Notice that the increase of $n_0$ decreases the number of (positive)
summation terms in $f(R)$, and affects $F^*$ (i.e., $F(\gamma^*)$)
in each term. Recall that Theorem~\ref{thm:gamma_R_n0_N} has shown
that $\gamma^*$ and thus $F^*$ are decreasing in $n_0$. Thus if we
can show that $f(R)$ is also increasing in $F^*$, then $f(R)$ is
decreasing in $n_0$.

The partial derivative of $f(R)$ over $F^*$ is
\begin{align}
&\frac{\partial f(R)}{\partial F^*}=(V-R)\cdot
\nonumber\\&\sum_{n=n_0}^N(n-NF^*)\left(\begin{array}{c}N\\n\end{array}\right)(F^*)^{n-1}(1-F^*)^{N-n-1}.
\label{eq:proof_pi_rho}\end{align} According to Theorem~\ref{them4},
the equilibrium collaborator number is
$$n^*=\sum_{n=0}^N n\left(\begin{array}{c}N\\n\end{array}\right)(F^*)^{n}(1-F^*)^{N-n},$$
which leads to $n^*=N\rho^*$. Thus we have
\begin{align}&(V-R)\sum_{n=0}^N(n-NF^*)\left(\begin{array}{c}N\\n\end{array}\right)(F^*)^{n-1}(1-F^*)^{N-n-1}\nonumber\\
&=\frac{V-R}{F^*(1-F^*)}\left(n^*-NF^*\right)=0.
\label{eq:proof_pi_rho_opt}\end{align} Notice that the sign of each
term in the summation operation of (\ref{eq:proof_pi_rho_opt}) is
decided by the relation between $n$ and $NF^*$, thus a term with
small $n$ is negative and a term with large $n$ is positive.
Compared to (\ref{eq:proof_pi_rho_opt}), $\partial f(R)/\partial
F^*$ in (\ref{eq:proof_pi_rho}) has less negative terms in the
summation operation and is thus positive. Thus we conclude that
$f(R)$ and equilibrium $f(R^*)$ are decreasing in $n_0$.

\section{Analysis of Model (B) in Three Information
Scenarios}\label{app:model2}

Here we turn to study Model (B) where the master will reward only with
successful collaboration. The analysis of this model is very similar
to Model (A), and in the following we briefly discuss the difference
between the two models due to the page limit.
\begin{itemize}
\item Under complete information, we can derive the same results as in
Theorem~\ref{thm:data_complete} for Model (B), by using a
similar analysis.

\item Under symmetrically incomplete information, for the equilibrium
of Stage II, we can similarly derive the same pure strategy NE as in
Theorem~\ref{them2} for Model (B), but the mixed strategy NE is different. The
mixed strategy NE exists only when $R$ is sufficiently large, and
the equilibrium probability $p^*$ in (3) is the unique solution to
\[
\mathbb{E}_m\left(\frac{R}{m+1}\boldsymbol{1}_{\{m+1\geq n_0\}}-\mu\right)=0,
\]
where the expectation $\mathbb{E}$ is taken over the random variable
$m$ that follows a binomial distribution $B(N-1,p)$. For the
equilibrium of the whole collaboration game, we can still derive the
same results as in Theorem~\ref{thm:data_symm}.
\item Under asymmetrically incomplete information, for the
equilibrium of Stage II, we can derive a similar equilibrium
decision threshold $\gamma^*(R)$ as the solution to
\begin{equation}\label{eq:proof_model2_asymm}
\mathbb{E}_m\left(\frac{R}{m+1}\boldsymbol{1}_{\{m+1\geq
n_0\}}-\gamma\right)=0,
\end{equation}
where the expectation is taken over $m$ that follows a binomial
distribution $B(N-1,F(\gamma))$.\footnote{Note that the solution to
(\ref{eq:proof_model2_asymm}) will exist only when $R$ is
sufficiently large, and the solution may not be unique. If there
exist two solutions (denoted by $\gamma_1^*$ and $\gamma_2^*$ with
$\gamma_1^*<\gamma_2^*$), each user $i$ will pick up $\gamma_2^*$
instead of $\gamma_1^*$ since it gives him a larger payoff
$\gamma_2^*-C_i$ (i.e., pareto-optimal for all users).} Then we can
similarly analyze the master's maximization problem in
(\ref{eq:asymm_f}). The difference from Model (A) is that here the
master needs to determine a larger reward $R$ to attract enough
users who face a higher risk.
\end{itemize}

\section{Proof of Theorem~\ref{thm:completeinfo}}\label{app:completeinfo}
\emph{Proof.} By observing Problem (\ref{eq:opt_symm_i}), the master
will only hire type-$i$ users when his marginal
{utility} is larger than marginal {cost (i.e.,
reward to users)} at $t_i=0$. That is,
$$\frac{d\pi_i(t_i)}{d t_i}\Big|_{t_i=0}=\left(\frac{N_i\theta_i}{1+N_it_i}-N_iK_i\right)|_{t_i=0}=N_i(\theta_i-K_i)>0,$$
{which does not depend on the other types.}
Thus the master will hire type-$i$ users only when $\theta_i>K_i$.
Since $\pi_i(t_i)$ is concave in $0\leq t_i\leq \bar{t}_i$, we can
directly examine the first-order condition of $\pi_i(t_i)$ over
$t_i$ for each type. Then we can derive the equilibrium contract
item for type-$i$ in (\ref{eq:contract_symm}).

By substituting all contract items into the objective function in
Problem~(\ref{eq:opt_symm}), we can further derive the master's
equilibrium profit in (\ref{eq:client_symm}).


\section{Proof of Proposition~\ref{prop:feasible}}\label{app:3contds}
\subsection{Proof of sufficient conditions}
We use mathematical induction to prove the three conditions in
Proposition~\ref{prop:feasible} are sufficient conditions for
contract feasibility. Let us denote $\mathcal{C}(l)$ as a subset
which contains the first $l$ task-reward combinations in the
contract $\mathcal{C}$. That is,
$\mathcal{C}(l)=\{(r_i,t_i)\}_{i=1}^l$.

We first show that $\mathcal{C}(1)$ is feasible. Since there is only
one user type, the contract is feasible as long as it satisfies IR
constraint for type-1. This is true due to $\mathtt{Condition}(+)$
in Proposition~\ref{prop:feasible}.

Next we show that if contract $\mathcal{C}(l)$ is feasible, then the
new contract $\mathcal{C}(l+1)$ by adding new item
$(r_{l+1},t_{l+1})$ is also feasible. To achieve this, we need to
show the following results.
\begin{itemize}
\item \emph{Result I:} the IC and IR constraints for
type-$({l+1})$ users:
\begin{equation}\label{eq:proof_result1}
\begin{cases} \; r_{l+1}-K_{l+1}t_{l+1}\geq r_{i}-K_{l+1}t_{i},\ \forall
i=1,...,l\\
\; r_{l+1}-K_{l+1}t_{l+1}\geq 0,
\end{cases}
\end{equation}
\item \emph{Result II:} for the original $l$ types already contained
in the contract $\mathcal{C}(l)$, the IC constraints are still
satisfied after adding the new type-$(l+1)$:
\begin{equation}\label{eq:proof_result2}
r_i-K_it_i\geq r_{l+1}-K_it_{l+1}, \forall i=1,...,l.
\end{equation}
Note that the new contract $\mathcal{C}(l+1)$ will satisfy the IR
constraints for all original $l$ types of users, since the original
contract $\mathcal{C}(l)$ is feasible.
\end{itemize}

\emph{Proof of Result I in (\ref{eq:proof_result1}):} First, we
prove the IC constraint for type-$(l+1)$. Since contract
$\mathcal{C}(l)$ is feasible, the IC constraint for a type-$i$ user
must hold, i.e.,
\[
r_j-K_lt_j\leq r_l-K_lt_l,\forall j=1,...,l.
\]
Also, the left inequality of (\ref{eq:contract_leq}) in
$\mathtt{Condition}(\leq)$ can be transformed to
\[
r_{l}+K_{l+1}(t_{l+1}-t_l)\leq r_{l+1}.
\]
By combining the above two inequalities, we have
\begin{equation}\label{eq:proof_model2_asymm_combine}
r_j-K_lt_j+K_{l+1}(t_{l+1}-t_l)\leq r_{l+1}-K_lt_l,\forall
j=1,...,l.
\end{equation}
Notice that $K_{l+1}<K_l$ and $t_j\leq t_l$ in
$\mathtt{Condition}(\uparrow)$, we also have
\[
K_{l+1}(t_l-t_j)\leq K_l(t_l-t_j).
\]
By substituting this inequality into
(\ref{eq:proof_model2_asymm_combine}), we have
\begin{equation}\label{eq:proof_model2_suff}
r_{l+1}-K_{l+1}t_{l+1}\geq r_j-K_{l+1}t_j,
\end{equation}
which is actually the IC constraint for type-$(l+1)$.

Next, we show that the IR constraint for type-$(l+1)$. Since
$K_{l+1}<K_j$ for any $j\leq l$, then
\[
r_j-K_{l+1}t_j\geq r_j-K_jt_j.
\]
By combining this inequality and (\ref{eq:proof_model2_suff}), we
have
\[
r_{l+1}-K_{l+1}t_{l+1}\geq r_j-K_jt_j\geq 0,
\]
due to the IR constraint for type-$j$. Thus we prove the IR
constraint for type-$(l+1)$ in (\ref{eq:proof_result1}).

\emph{Proof of Result II in (\ref{eq:proof_result2}):} Since
contract $\mathcal{C}(l)$ is feasible, the IC constraint for
type-$j$ holds, i.e.,
\[
r_l-K_jt_l\leq r_j-K_jt_j,\forall j=1,...,l.
\]
Also, we can transform the right inequality of
(\ref{eq:contract_leq}) in $\mathtt{Condition}(\leq)$ to
\[
r_{l+1}\leq r_l+K_l(t_{l+1}-t_l).
\]
By combining the above two inequalities, we conclude
\[
r_{l+1}-K_jt_l\leq K_l(t_{l+1}-t_l)+r_j-K_jt_j.
\]
Notice that $K_l<K_j$ and $t_{l+1}\geq t_l$ in
$\mathtt{Condition}(\uparrow)$, we also have
\[
K_l(t_{l+1}-t_l)\leq K_j(t_{l+1}-t_l).
\]
By combining the above two inequalities, we conclude
\[
r_j-K_jt_j\geq r_{l+1}-K_jt_{l+1},\forall j=1,...,l,
\]
which is actually the IC constraint for type-$j$ in
(\ref{eq:proof_result2}).

\subsection{Proof of necessary conditions}
We prove the three conditions in Proposition~\ref{prop:feasible}
are necessary conditions for contract feasibility. It is easy to see
that $\mathtt{Condition}(+)$ is just the IR condition for type-$1$
in a feasible contract. Also, the right inequality of
$\mathtt{Condition}(\leq)$ can be derived from the IC constraint for
type-$(i-1)$, and the left inequality can be derived from the IC
constraint for type-$i$.

Next we prove $\mathtt{Condition}(\uparrow)$ is also the necessary
condition. We divide the proof into two parts.
\begin{itemize}
\item We first prove that if $K_i>K_j$ then $t_i\leq t_j$ by
contradiction. Suppose $t_i>t_j$, then we have
\begin{equation}\label{eq:proof_model2_nec}
K_i(t_i-t_j)>K_j(t_i-t_j),
\end{equation}
due to $K_i>K_j$. Notice that the feasible contract satisfies the IC
constraints for type-$i$ and type-$j$ users, we have
\[
r_i-K_it_i\geq r_j-K_it_j,
\]
and
\[
r_j-K_jt_j\geq r_i-K_jt_i.
\]
By combining the above two inequalities, we conclude
\[
K_it_i+K_jt_j\leq K_it_j+K_jt_i,
\]
which contradicts with (\ref{eq:proof_model2_nec}).
\item We then prove that $t_i\geq t_j$ if and only if $r_i\geq r_j$.
\begin{itemize}
\item If $t_i>t_j$, we want to prove $r_i>r_j$. Due to the IC
constraint for type-$i$, we have
\[
r_i-K_it_i\geq r_j-K_it_j,
\]
which can be transformed to
\[
r_i-r_j\geq K_i(t_i-t_j).
\]
Since $t_i>t_j$, we can derive $r_i>r_j$ from the above inequality.
\item If $r_i>r_j$, we want to prove that $t_i>t_j$. Due to the IC
constraint for type-$j$, we have
\[
r_j-K_jt_j\geq r_i-K_jt_i,
\]
which can be transformed to
\[
K_j(t_i-t_j)\geq r_i-r_j.
\]
Since $r_i>r_j$, we can derive $t_i>t_j$ from the above inequality.
\item Using a similar analysis, we can prove that $r_i=r_j$ if and
only if $t_i=t_j$.
\end{itemize}
\end{itemize}

\section{Proof of Theorem~\ref{thm:incompleteinfo}}\label{app:incompleteinfo}
All {involved users} in set $\mathcal{I}_A$ will receive positive rewards and
tasks.  According to $\mathtt{Condition} (\uparrow)$, the rewards
and tasks are {non-decreasing} in the types. Let us denote the lowest
type of involved users in set $\mathcal{I}_A$ as type-$\hat{\j}$.
{If $\hat{\j}=1$, then relation~(\ref{eq:opt_r_t1}) shows that a
type-$1$ collaborator receives a zero payoff. If $\hat{\j}>1$, then
any lower type $k<\hat{\j}$ is not in set $\mathcal{I}_A$, and
receives zero task and zero reward. By using
relation~(\ref{eq:opt_r_ti}), we can further derive that
$r_{\hat{\j}}^*=K_{\hat{\j}}t_{\hat{\j}}^*$, which means the lowest
type collaborator still obtains a zero payoff.}

{According to (\ref{eq:opt_r_ti}), the
type-$i$ collaborator's equilibrium payoff is
$r_i^*-K_it_i^*=r_{i-1}^*-K_it_{i-1}^*$, which is strictly larger
than type-$(i-1)$ collaborator's payoff $r_{i-1}^*-K_{i-1}t_{i-1}^*$
as $K_i<K_{i-1}$.} {Thus a higher type collaborators receive a
larger positive payoff.}

Next we show which types of users are involved as collaborators. By observing the
first derivative of the master's expected profit over $t_i$ in (\ref{eq:der_E_pi}), $t_i$ only appears in the last bracket. The master will
involve type-$i$ users only when the last bracket of
(\ref{eq:der_E_pi}) is positive at $t_i=0$. {This leads to} the
collaborator set in (\ref{eq:collaborate_set_asymm}). By comparing
$\mathcal{I}_C$ in (\ref{eq:collaborate_set_symm}) and
$\mathcal{I}_A$ in (\ref{eq:collaborate_set_asymm}), we conclude
{that} $|\mathcal{I}_A|\leq |\mathcal{I}_C|$.

\end{document}